\documentclass{llncs}

\newcommand{\norm}[1]{}
\usepackage[utf8]{inputenc}
\usepackage[mathscr]{eucal}
\usepackage{graphicx}
\usepackage{listings}
\usepackage{amssymb}
\usepackage{dsfont}
\usepackage{mathtools}
\usepackage{algorithm}
\usepackage{algpseudocode}
\usepackage{xcolor}
\usepackage{multirow}
\usepackage{amsmath}
\usepackage[english]{babel}
\usepackage{wasysym}
\usepackage{qip}
\usepackage{qm}
\usepackage{csquotes}
\usepackage{environ}

\usepackage[draft=false]{hyperref}
\usepackage[normalem]{ulem}
\usepackage[capitalise]{cleveref}
\crefname{construct}{Construction}{Constructions}

\usepackage{comment}
\usepackage{paralist}
\usepackage{wrapfig}
\usepackage[normalem]{ulem}

\usepackage[clock]{ifsym}

\usepackage[
	n,
	operators,
	advantage, 
	sets,
	adversary,
	landau,
	probability,
	notions,
	logic,
	ff,
	mm,
        primitives,
        events,
        complexity,
        asymptotics, 
        keys]{cryptocode}

\usepackage[maxbibnames=10]{biblatex}

\spnewtheorem{Claim}{Claim}{\bfseries}{\itshape}

\spnewtheorem{construct}{Construction}{\bfseries}{\itshape}
\spnewtheorem*{thm}{Theorem}{\bfseries}{\itshape}
\spnewtheorem*{lem}{Lemma}{\bfseries}{\itshape}
\definecolor{darkblue}{rgb}{0,0,0.6}
\definecolor{darkgreen}{rgb}{0,0.5,0}

\newcommand{\mo}[1]{}
\newcommand{\takashi}[1]{}
\newcommand{\ryo}[1]{}
\newcommand{\mor}[1]{}

\newcommand{\D}{\ensuremath{\mathcal{D}}}
\newcommand{\Dis}{\ensuremath{\mathbf{D}}}

\usepackage{xspace}

\newcommand{\QMA}{\mathsf{QMA}}

\newcommand{\QCMA}{\mathsf{QCMA}}

\newcommand{\BQP}{\mathsf{BQP}}

\newcommand{\PQCMA}{\mathsf{QCMA}}

\newcommand{\V}{\mathcal{V}}
\newcommand{\Pv}{\mathcal{P}}

\newcommand{\Lg}{\mathcal{L}}

\newcommand{\G}{\mathcal{G}}

\newcommand{\NICV}{\mathsf{NI}\text{-}\mathsf{CVQC}}

\newcommand{\Hy}{\ensuremath{\textbf{H}}}

\newcommand{\Sim}{\ensuremath{\mathcal{S}}}

\begin{document}
\title{Separating Non-Interactive Classical Verification of Quantum Computation from Falsifiable Assumptions}

\author{Mohammed Barhoush\inst{1}\thanks{Part of this work was done while visiting NTT Social Informatics Laboratories as an internship.} \and Tomoyuki Morimae\inst{2} \and Ryo Nishimaki\inst{3}  \and Takashi Yamakawa\inst{3,2}}
\institute{Universit\'e de Montr\'eal (DIRO), Montr\'eal, Canada\\  \email{mohammed.barhoush@umontreal.ca} \and Yukawa Institute for Theoretical Physics, Kyoto University, Kyoto, Japan\\ \email{tomoyuki.morimae@yukawa.kyoto-u.ac.jp} \and
NTT Social Informatics Laboratories\\
\email{ryo.nishimaki@ntt.com}, \email{takashi.yamakawa@ntt.com}}%

\maketitle

\begin{abstract}
Mahadev [SIAM J. Comput. 2022] introduced the first protocol for classical verification of quantum computation based on the Learning-with-Errors (LWE) assumption, achieving a 4-message interactive scheme. This breakthrough naturally raised the question of whether fewer messages are possible in the plain model. Despite its importance, this question has remained unresolved.

In this work, we prove that there is no quantum black-box reduction of non-interactive classical verification of quantum computation of $\QMA$ to any falsifiable assumption. Here, “non-interactive” means that after an instance-independent setup, the protocol consists of a single message.
\mor{Maybe adding something like as follows? (Here,  we use the term ``non-interactive'' to describe two-message protocols where the first message is independent of the statement to be proven.)}\takashi{If we regard it as a special case of two-message protocol, then it would be more natural to consider non-adaptive soundness since the instance is usually fixed before starting the protocol. Thus, I feel that adaptive soundness considered in the paper only makes sense when we regard it as a non-interactive protocol after an instance-independent setup.
}\mor{My intention of that comment was that I feel it seems better to mention that our "non-interactive" is not really "non-interactive". People might confuse that we consider the case where the prover sends a message and just it.}\mo{I added a statement to clarify above.}\takashi{My concern is that if we explain in that way, then that would give a wrong impression that we are considering non-adaptive soundness.
How about the following? "Here, “non-interactive” means that after an instance-independent setup, the protocol consists of a single message."? 
}\mo{ok I used your sentence.}
This constitutes a strong negative result given that falsifiable assumptions cover almost all standard assumptions used in cryptography, including LWE.  Our separation holds under the existence of a $\QMA \text{-} \QCMA$ gap problem. Essentially, these problems require a slightly stronger assumption than $\QMA\neq \QCMA$. To support the existence of such problems, we present a construction relative to a quantum unitary oracle.
\end{abstract}

\newpage
\section{Introduction}


Quantum computation and communication promise capabilities that are beyond what is classically achievable. Yet, building and operating quantum devices remains expensive and technologically demanding. A natural and practically important solution is delegated quantum computation: classical users outsource computations to powerful quantum servers, while retaining the ability to verify correctness. Designing such verification protocols is a central challenge in quantum cryptography. 


Early work in this direction \cite{B15,MF16,DBE+17} sought to minimize the quantum resources required of the verifier. A landmark breakthrough came with Mahadev’s protocol \cite{M18}, which showed for the first time that a fully classical verifier can reliably certify the outcome of quantum computations performed by an untrusted quantum prover. More concretely, Mahadev constructed an interactive proof system between a classical probabilistic polynomial-time (PPT) verifier and a quantum polynomial-time (QPT) prover with the following guarantee: for any language $L \in \QMA$ \footnote{Mahadev's work focuses on $\BQP$, but her $\textsf{CVQC}$ construction can be easily generalized to $\QMA$.}, if $x \in L$ with witness $\ket{\psi_x}$, the prover with input $x$ and polynomial copies of $\ket{\psi_x}$, can convince the verifier of this fact, while if $x \notin L$, no malicious QPT prover can succeed in convincing the verifier, except with negligible probability.

Mahadev’s protocol requires four messages of interaction and relies on the quantum hardness of the Learning-with-Errors (LWE) problem \cite{R07}. Subsequent works have refined various aspects of the original scheme \cite{GV19, Z22, NZ23, BKM+25, ACG+20, CCY20, B21}. In particular, several lines of research reduced the round complexity in the quantum random oracle model \cite{ACG+20,CCY20,B21}, ultimately achieving a non-interactive scheme \cite{ACG+20}. Despite these advances, reducing the round complexity in the plain model has remained a central and long-standing open problem in quantum cryptography. 

\subsection{Our Contribution} 

In this work, we show that non-interactive {classical verification of quantum computation} $(\NICV)$ of $\QMA$ cannot be based on any falsifiable assumptions, assuming the existence of a notion termed $\QMA \text{-} \QCMA$ gap problem. Following \cite{ACG+20}, we use the term ``non-interactive'' to describe two-message protocols where the first message is independent of the statement to be proven \footnote{Such a message is often referred to as a common reference string or public key.}. 

More precisely, a $\NICV$ for a language $L\in \QMA$ is an interactive protocol between a QPT prover and PPT verifier, specified by three algorithms $(\mathcal{G},\Pv,\mathcal{V})$ with the following syntax:
\begin{enumerate}
    \item \textbf{Key generation:} On input the security parameter $1^n$, the verifier runs $(\sk,\pk)\gets \mathcal{G}(1^n)$ to produce a public key $\pk$ and a secret verification key $\sk$. The verifier sends $\pk$ to the prover. (This step is independent of the instance.)
    \item \textbf{Proof generation}: Given the public key $\pk$, an instance $x \in L$, and a quantum witness state $\ket{w_x}^{\otimes t}$ for $x$ with an appropriate number of copies $t$, the prover computes a classical proof $\pi\gets \mathcal{P}(\pk,x,\ket{w_x}^{\otimes t})$ and sends $(x,\pi)$ to the verifier.
    \item \textbf{Verification:} Using $\sk$, the verifier runs $1/0\gets \mathcal{V}(\sk,x,\pi)$ to decide whether to accept the claim that $x\in L$.
\end{enumerate}
Soundness requires that no QPT malicious prover can convince the verifier to accept an element not in $L$ {except with negligible probability} \footnote{In fact, our impossibility result also excludes schemes with a more relaxed notion of soundness, where an adversary succeeds with at most a constant $p<1$ probability.}. 
Note that our definition and impossibility result covers $\NICV$ with designated-verifiers, meaning that the verifier is allowed to use a secret verification key $\sk$, which is hidden from the prover. 


Our main result is a quantum black-box separation: 
\begin{theorem}[Informal of \cref{thm:main}]
\label{inf thm: main}
    Let $L$ be a $\QMA$ language with a subexponential $\QMA \text{-} \QCMA$ gap problem and let $\Pi$ be a $\NICV$ for $L$. Then, for any falsifiable assumption, one of the following statements hold:
\begin{enumerate}
    \item The falsifiable assumption is false.
    \item There is no quantum black-box reduction showing the soundness of $\Pi$ from the falsifiable assumption. 
\end{enumerate}
\end{theorem}

Falsifiable assumptions encompass essentially all standard cryptographic assumptions, including one-way functions, trapdoor permutations, RSA, and LWE. Intuitively, an assumption is falsifiable if it can be cast as an efficient interactive game between a challenger and an adversary, where the challenger can verify whether the adversary has won the game. Note that we cannot establish a separation from non-falsifiable assumptions, as $\NICV$ is itself a non-falsifiable assumption. By \cref{inf thm: main} and Mahadev's 4-message protocol for $\textsf{CVQC}$, we immediately obtain the following corollary: 
\begin{corollary}
    Assuming LWE and the existence of a subexponential $\QMA \text{-} \QCMA$ gap problem, there is no quantum black-box reduction of $\NICV$ to 4-message $\textsf{CVQC}$. 
\end{corollary}

A technical component of our separation is the notion of $\QMA \text{-} \QCMA$ gap problems. Informally, this is a $\QMA$ problem with the following properties:
\begin{itemize}
    \item Yes instances, together with valid witnesses, can be efficiently sampled.
    \item No instances can be (possibly inefficiently) sampled.
    \item No QPT adversary with access to a $\PQCMA$ oracle can distinguish between yes and no instances.
\end{itemize}
Note that, in this work, a $\QCMA$ oracle refers to an oracle that solves any Promise $\QCMA$ problem (see \cref{sec:lang}). In the subexponential version of a $\QMA \text{-} \QCMA$ gap problem, indistinguishability holds even against subexponential-time adversaries. The existence of such problems is deeply connected to the longstanding open question of whether $\QMA$ contains problems that do not belong to $\QCMA$—a fundamental question explored in several studies \cite{AK07,LLP+23,NN23,LMY24,Z24,BHN+25,BHV26}. Such works demonstrate separations between $\QMA$ and $\QCMA$ relative to various oracles, thus providing supporting evidence for the existence of a $\QMA \text{-} \QCMA$ gap problem. 

That being said, a $\QMA \text{-} \QCMA$ gap problem requires slightly more than $\QMA\neq \QCMA$, so we provide further evidence by giving a construction in this work. 

\begin{theorem}[Informal of \cref{result 2}]
There exists a subexponential $\QMA \text{-} \QCMA$ gap problem relative to a quantum unitary oracle. 
\end{theorem}

\subsection{Limitations and Open Problems}

Our separation holds for $\NICV$ of $\textsf{QMA}$, and not $\textsf{BQP}$. Indeed, it is not possible to separate $\textsf{CVQC}$ for $\BQP$ from falsifiable assumptions: the primitive itself is a falsifiable assumption as checking whether an element belongs to a $\BQP$ language can be done in QPT. 

That said, we expect the proof technique to adapt relatively straightforwardly to yield a conditional impossibility of \emph{classical} reductions of $\NICV$ for $\BQP$ to falsifiable assumptions with \emph{classical} challengers/verifiers. A natural approach is to use a suitable gap problem that remains secure against $\textsf{NP}$ distinguishers---that is, against classical polynomial-time algorithms with access to an \textsf{NP} oracle. Provided there exists a $\BQP$ language that is hard for such $\textsf{NP}$ distinguishers, then an analogous argument should establish the desired separation. Formalizing this approach may be an interesting direction of future work. 


Furthermore, $\NICV$ are naturally required to satisfy adaptive soundness, where the cheating prover may choose the false statement after seeing the public key $\pk$. In contrast, under the weaker static soundness notion, the false statement is chosen before the verifier sends its public key. The latter notion is not covered by our impossibility result. Extending our separation to rule out static soundness remains an interesting open problem (or, alternatively, constructing such arguments would also be significant).


\subsection{Related Work}

\subsubsection{Black-Box Separation for \textsf{SNARG}s.}

Our separation shares structural similarities with the seminal result of Gentry and Wichs \cite{GW11}, who established a black-box separation between succinct non-interactive arguments (\textsf{SNARG}s) and falsifiable assumptions. A \textsf{SNARG} is a non-interactive proof system for \textsf{NP} statements in which the message lengths are \emph{short}--traditionally at most polylogarithmic in the instance and witness length. Gentry and Wichs introduced the notion of \emph{subexponential subset membership problems} for \textsf{NP}, requiring an \textsf{NP} language with efficiently samplable yes instances (with witnesses) and not-necessarily-efficiently samplable no instances such that the two distributions are indistinguishable against subexponential-time adversaries. Their work showed that no black-box reduction can base the soundness of \textsf{SNARG}s on falsifiable assumptions assuming the existence of such problems. This separation was later extended to the quantum setting, ruling out even quantum black-box reductions \cite{ADS+24}. 


\subsubsection{Separating $\QMA$ and $\QCMA$.}

The existence of $\QMA \text{-} \QCMA$ gap problems is closely tied to the broader question of separating $\QMA$ from $\QCMA$—a central open problem in quantum complexity theory. Aaronson and Kuperberg \cite{AK07} first established a quantum unitary oracle relative to which $\QMA\neq \QCMA$. Ideally, such separations should rely on classical oracles, leading to several follow-up works moving towards this goal \cite{LLP+23,NN23,LMY24,Z24,BHN+25}. 

Our construction of a $\QMA \text{-} \QCMA$ gap problem is based on a quantum unitary oracle, resembling Aaronson and Kuperberg's oracle. However, our approach requires additional structure: we need a $\QMA$ language where yes and no instances are indistinguishable against $\QCMA$ oracle-aided distinguishers, rather than merely undecidable by $\QCMA$ algorithms.  We believe that techniques from more recent oracle separations \cite{LLP+23,NN23,LMY24,Z24,BHN+25} may also be adapted to construct $\QMA \text{-} \QCMA$ gap problems in alternative oracle models. Exploring this direction is a potential avenue for future research.

\subsection{The Gentry-Wichs Separation.} 
Since our separation bears some similarities to the Gentry-Wichs separation of \textsf{SNARG}s from falsifiable assumptions \cite{GW11}, we briefly review their construction.

Let $L$ be an \textsf{NP} language with a subexponential gap problem. Assume that there exists a black-box reduction $\Sigma$ from a \textsf{SNARG} to a falsifiable assumption. Then, by definition, $\Sigma$ must succeed in breaking the assumption when given oracle access to an adversary breaking the \textsf{SNARG}.

At a high level, the separation proceeds in four steps:

\begin{enumerate}
    \item \textit{The Leakage Lemma.} If two distributions $\Lg$ and $\overline{\Lg}$ are indistinguishable, then even when augmented with auxiliary strings of short length, their extensions remain indistinguishable. The shortness is key: it ensures that brute-force search suffices to recover auxiliary strings, so this information should not allow distinguishing yes and no instances of the language. This result is known as the leakage lemma. 
    

    \item \textit{Constructing an adversary.} 
    Define $\Lg$ and $\overline{\Lg}$ as yes/no instance distributions of $L$. Extend $\Lg$ to $\Lg^*$ by pairing instances with their honest \textsf{SNARG} proofs. The leakage lemma implies the existence of a distribution $\overline{\Lg}^*$ over no instances paired with valid-looking proofs that is indistinguishable from $\Lg^*$. Note that these proofs must pass the verification algorithm of the $\textsf{SNARG}$ scheme, otherwise the two distributions, ${\Lg}^*$ and $\overline{\Lg}^*$, can easily be distinguished. Let $\overline{\Pv}$ be the possibly-inefficient algorithm that samples from $\overline{\Lg}^*$. 
    

    \item \textit{Simulating the adversary.} Let $\Sim$ be the algorithm that samples yes instances and honestly proves them. In other words, $\Sim$ samples from the distribution $\Lg^*$. Use sample indistinguishability of $\Lg^*$ and $\overline{\Lg}^*$ to argue oracle indistinguishability of $\Sim$ and $\overline{\Pv}$. 
    
    \item \textit{Breaking the reduction.} Given that $\overline{\Pv}$ produces valid proofs for no instances, it constitutes an inefficient attack against \textsf{SNARG} soundness. Thus, $\Sigma^{\overline{\Pv}}$ constitutes an attack against the assumption. Because $\Sigma$ cannot distinguish oracle access to $\overline{\Pv}$ from $\Sim$, it must succeed in breaking the assumption in both cases. But since $\Sim$ is efficient, $\Sigma^{\Sim}$ yields an efficient attack on the assumption—a contradiction.
\end{enumerate}

\subsection{Technical Overview of Our Separation.}

We now describe our separation of $\NICV$ from falsifiable assumptions. Let $L$ be a \textsf{QMA} language with a subexponential $\QMA\text{-}\QCMA$ gap problem. Assume that there exists a quantum black-box reduction $\Sigma$ from a $\NICV$ to a falsifiable assumption. 

\subsubsection{Modified leakage lemma.}
Let $\Lg$ and $\overline{\Lg}$ be yes and no instance distributions of $L$. We aim to establish an analogous result to the leakage lemma, showing that $\Lg$ and $\overline{\Lg}$ remain indistinguishable in the presence of auxiliary classical information, reflecting the possible proofs in a $\NICV$. However, in our case, the  $\NICV$ proofs are not necessarily short; so the resulting space of possible proofs is exponentially large. This renders brute-force search computationally intractable, meaning the leakage lemma does not apply.  


To address this challenge, we augment the adversary with access to a $\PQCMA$ oracle. This provides a carefully calibrated level of computational power: it is sufficient to facilitate the search without compromising the hardness of the underlying $\QMA$ problem. More specifically, the adversary must be able to search through the space of possible proofs (as in the Gentry-Wichs scheme), without gaining the capacity to break the indistinguishability between $\Lg$ and $\overline{\Lg}$. The ability to ``search through classical proofs'' is similar to the ability to solve $\PQCMA$ problems since this involves finding a classical witness, so this oracle is the natural choice. 

Roughly, a $\PQCMA$ oracle takes an input $(V,x)$, consisting of a QPT algorithm $V$ and an input $x$, and outputs 1 if there exists a classical witness $w$ such that $\Pr[V(x,w)=1]\geq 2/3$ and outputs $0$ if for any string $\widetilde{w}$ of appropriate length, $\Pr[V(x,\widetilde{w})=1]\leq 1/3$. Our key insight is that finding the correct auxiliary information can be modeled as a classical witness search,  making it reducible to a $\PQCMA$ problem. On the other hand, solving a $\QMA$ problem requires a \emph{quantum} witness; thus, $\QMA$ may remain hard even for an adversary with $\PQCMA$ oracle access. We formalize this by asserting that the $\QMA$ language has a $\QMA \text{-} \QCMA$ gap problem.

Leveraging this oracle requires a more sophisticated approach than the Gentry-Wichs brute-force method, as casting the search for auxiliary information into a $\PQCMA$ framework is non-trivial. 

With this idea, we establish our version of the leakage lemma. 

\begin{theorem}[Informal of \cref{result 1}]
Let $p$ be a polynomial and let $L$ be a language with a $\QMA \text{-} \QCMA$ gap problem $(\mathcal{L},\overline{\mathcal{L}})$ (yes, no instance distributions). Assume that there exists a distribution $\mathcal{L}^*$ on pairs $(x,\pi)$ where $x$ is distributed according to $\mathcal{L}$ and $\pi$ is some arbitrary string of length at most $p(\lvert x\rvert)$. 
Then there exists a distribution $\overline{\Lg}^*$ on pairs $(\overline{x},\overline{\pi})$ where $\overline{x}$ is distributed according to 
$\overline{\Lg}$ and $\overline{\pi}$ has length at most $p(\lvert \overline{x}\rvert)$ such that $\Lg^*$ and $\overline{\Lg}^*$ are computationally indistinguishable.
\end{theorem}

\subsubsection{Upgrading sample-indistinguishability to oracle indistinguishability.}
We apply our modified leakage lemma to the distribution $\Lg^*$, which consists of yes instances of a $\QMA$ language alongside honestly generated $\NICV$ proofs. The lemma then guarantees the existence of a distribution $\overline{\Lg}^*$ consisting of pairs $(\overline{x},\overline{\pi})$ that is indistinguishable from $\Lg^*$. 

In the Gentry-Wichs framework, the subsequent step involves constructing a prover $\overline{\Pv}$ and a simulator $\Sim$ that sample from $\Lg^*$ and $\overline{\Lg}^*$, respectively, to argue that the reductions $\Sigma^{\overline{\Pv}}$ and $\Sigma^{\Sim}$ are indistinguishable. However, since we are trying to rule out \emph{quantum} black-box reductions, we must consider reductions $\Sigma$ that possess oracle access to unitary versions of $\overline{\Pv}$ and $\Sim$. This is problematic since our modified leakage lemma only shows that a classical sample from $\Lg^*$ and from $\overline{\Lg}^*$ are indistinguishable. Therefore, we need to upgrade this sample indistinguishability to ensure that even quantum oracle access is insufficient to distinguish the two distributions. 

A similar issue was faced in \cite{ADS+24} when generalizing Gentry-Wichs to rule out quantum black-box reductions for $\textsf{SNARG}$s to falsifiable assumptions. However, their context was simpler: the prover in a $\textsf{SNARG}$ is a polynomial-time classical algorithm. This allows for the simulation of quantum oracle access using a set of classical samples. In our case, the honest prover in a $\NICV$ is a \emph{quantum} algorithm. Hence, it is not clear how to simulate oracle access to a quantum prover using classical samples. To resolve this issue, we employ the compressed oracle technique from \cite{TZ25,Z19}, allowing us to build a (stateful) QPT algorithm $\Sim$ that is quantum oracle indistinguishable from $\overline{\Pv}$. 

\subsubsection{Separating $\NICV$ from falsifiable assumptions.}

Finally, we follow quantum version of Gentry-Wichs \cite{ADS+24} to separate $\NICV$ from falsifiable assumptions as follows. By the previous argument, we showed that no QPT algorithm can distinguish quantum oracle access to $\overline{\Pv}$ and $\Sim$. Since the assumed black-box reduction $\Sigma$ is a QPT algorithm, the output of $\Sigma^{\overline{\Pv}}$ and $\Sigma^{\Sim}$ should be indistinguishable. 

On the one hand, since $\overline{\Pv}$ produces valid proofs for no instances, thus breaking $\NICV$ soundness, $\Sigma^{\overline{\Pv}}$ successfully breaks the falsifiable assumption. On the other hand, $\Sigma^{\Sim}$ is indistinguishable from $\Sigma^{\overline{\Pv}}$ and, thus, also breaks the falsifiable assumption by indistinguishability. This gives a contradiction since $\Sigma^{\Sim}$ is efficient, given that $\Sim$ is efficient, thus, implying the existence of an efficient attack against the falsifiable assumption.   

\subsection{Building a $\QMA \text{-} \QCMA$ gap problem}

Notice that the attack presented above only applies to $\QMA$ languages with a $\QMA \text{-} \QCMA$ gap problem. In order to support the existence of such a problem, we give a construction relative to a quantum unitary oracle. We use a result from \cite{AK07} (\cref{lem:oracle} in our paper), which states that no QPT adversary with the aid of a classical witness can distinguish oracle access to the identity unitary from a unitary that adds a phase $-1$ to some random Haar state and acts as the identity on all other orthogonal queries. 

We construct a (subexponential) $\QMA \text{-} \QCMA$ gap problem as follows. We sample, for each $n\in \mathbb{N}$, a language $L_n\subseteq\{0,1\}^n$ uniformly at random, and define $L=\cup_n L_n$. Then, for each element $x$ in the language, we sample a Haar random state $\ket{\psi_x}$. The oracle $U_L$ is defined to apply a phase of $(-1)$ on input states of the form $\ket{x}\ket{\psi_x}$ for $x\in L$, while acting as the identity on all orthogonal inputs. Therefore, the only way to detect whether $x\in L$ is to possess the specific state $\ket{\psi_x}$, which intuitively by the Haar-randomness and \cref{lem:oracle} appears completely random to any adversary with classical advice in the form of a $\PQCMA$ oracle. 

While the intuition is straightforward, proving it requires addressing some technical hurdles that go beyond the scope of \cite{AK07}. Our case deals with a distinguisher that has oracle access to a $\QCMA$ solver instead of a classical witness, as well as oracles that sample from yes and no distributions. As such, the adversary potentially has more power than in \cite{AK07}. Specifically, the $\QCMA$ oracle allows for adaptive access to a $\QCMA$ oracle, which is possibly stronger than obtaining a classical witness as in \cite{AK07}. Furthermore, the sampling oracles leak information which may aid in the attack.

To resolve these issues, we first present a result (\cref{thm:oracle}) showing that if yes and no instances of a language are indistinguishable against distinguishers with classical witnesses, then the language remains indistinguishable against distinguishers with $\QCMA$ oracle access. This allows us to reduce the complexity of the $\QCMA$ oracle to a more manageable classical witness model. We then employ a hybrid argument to systematically decouple the sampling oracles from the secret states. Finally, by applying the reduced result from \cite{AK07}, we prove the existence of a $\QMA \text{-} \QCMA$ gap problem relative to our quantum unitary oracle.




\section{Preliminaries}
\label{sec:prelim}

\subsection{Notation}
\label{sec:notation}

We often denote sets as $X,Y$, distributions with bold such as $\textbf{H}, \textbf{D}$, and algorithms as $\mathcal{A},\mathcal{D}$. We implicitly assume that all sets and distributions in this work are parameterized by some integer $n$, typically a security parameter. We write $x\leftarrow \textbf{D}$ to mean that $x$ is sampled according to the distribution $\textbf{D}$. If $X$ is a set, then $x\leftarrow X$ means that $x$ is chosen uniformly at random from the set. 
Let $[n]\coloneqq \{1,2,\ldots,n\}$ for every $n\in \NN$. Furthermore, let $\negl[n]$ denote any function that is asymptotically smaller than the inverse of any polynomial. 
If $\Dis$ is a distribution over a set $Y$ and $X$ is a set, the notion $f\leftarrow \Dis^X$ denotes sampling a function $f$ mapping $X$ to $Y$ such that for every $x\in X$, $f(x)$ is sampled according to the distribution $\Dis$. 

We follow the standard notation of quantum information~\cite{NC00}. $I$ is the two-dimensional identity operator. For the notational simplicity, we often denote $I^{\otimes n}$ by $I$ if the dimension is clear from the context.
Let $\mu_n$ denote the {Haar measure} over $n$-qubit pure-states. 
We say that a quantum algorithm $\adv$ is \emph{$s$-time}, for some function $s:\mathbb{N}\rightarrow \mathbb{N}$, if for any input $x$, $\adv(x)$ has run-time at most $s(\lvert x\rvert)$. If $s$ is a polynomial, then we say that $\adv$ is a \emph{QPT} algorithm.
An algorithm $\adv$ is \emph{non-uniform} if it is initialized with a classical (possibly randomized) advice and is said to be \emph{uniform} if no advice is given. Adversaries in the security definitions are assumed to be non-uniform unless mentioned otherwise. 
For a unitary $U$, we write $\adv^U(x)$ to mean that the algorithm $\adv$ has quantum query access to $U$. For an integer $q$, we write $\adv^{qU}$ to mean that $\adv$ only has access to $q$ queries to $U$. Similarly, for a function $F$, $\adv^F$ denotes quantum query access to a unitary implementation $U_F$ of $F$ defined as the map $U_F:|x\rangle|y\rangle\to|x\rangle|F(x)\oplus y\rangle$,
and let $\adv^{cF}$ denote \emph{classical} oracle access to $F$.

\subsection{Compressed Oracle}
\label{sec:compressed}

We recall a tool used to simulate oracle access to a certain set of classical functions. Consider a set of distributions denoted by $\{\Dis_x\}_{x \in \{0,1\}^*}$. Let $\Dis := \bigotimes_x \Dis_x$. 

An oracle $O$ is considered to be drawn from $\Dis$, i.e. $O\gets \Dis$, if every row $x$ in its associated truth table is sampled independently according to $\Dis_x$. 

The following result describes how to simulate access to an oracle sampled from $\Dis$ and is given as Theorem 7.1 in \cite{TZ25}. Note that Theorem 7.1 is more general as it involves a state $\ket{\psi}$ which is used to sample from $\Dis$. However, we do not require such a state for our applications.

\begin{theorem}\label{thm:compressed}
    Let $q=q(n)$ and $m=m(n)$ be polynomials on the security parameter $n\in \mathbb{N}$. Let $\Dis\coloneqq \{\Dis_x\}_{x\in \{0,1\}^m}$ be a tuple of distributions such that for any $x\in \{0,1\}^m$, $\Dis_x$ can be sampled  in QPT in $m$. Then, there exists a stateful QPT unitary oracle $\Sim_\Dis$ such that for any $q$-query quantum algorithm $\adv$, 
\begin{align}
\Pr[\mathcal{A}^{\Sim_\Dis}(1^n) = 1] = \Pr_{O \leftarrow \Dis} [\mathcal{A}^O (1^n)= 1].
\end{align}
\end{theorem}


\subsection{Quantum Complexity Classes}
\label{sec:lang}

We recall the definitions of the complexity classes, Quantum-Merlin-Arthur $(\QMA)$ and Quantum-Classical-Merlin-Arthur $(\QCMA)$. 

\begin{definition}[$\QMA$]
    A language $L\subseteq \{0,1\}^*$ is in $\QMA$ if there exists a QPT algorithm $\mathcal{Q}$ and a polynomial $p$ such that
    for any $x\in\{0,1\}^*$,
    \begin{enumerate}
        \item If $x\in L$, then there exists a $p(\lvert x\rvert)$-qubit state $\ket{w_x}$ such that 
        \begin{align}
        \Pr[1\gets\mathcal{Q}(x,\ket{w_x})]\geq \frac{2}{3}. 
        \end{align}
        \item If $x\notin L$, then for any $p(|x|)$-qubit state $\ket{w}$,  
        \begin{align}
        \Pr[1\gets\mathcal{Q}(x,\ket{w})]\leq \frac{1}{3}. 
        \end{align}
    \end{enumerate}
\end{definition}

\begin{definition}[$\QMA$ Relation]
Let $L$ be a language in $\QMA$ and let $\mathcal{Q}$ be its verification algorithm. 
A $\QMA$ relation $\mathcal{R}_{L,\mathcal{Q}}$ is the set of pairs 
$(x,\ket{w})$ satisfying $\Pr[1\gets\mathcal{Q}(x,\ket{w})]\ge\frac{2}{3}$.
\end{definition}


\begin{definition}[$\QCMA$]
    A language $L\subseteq\{0,1\}^*$ is in $\QCMA$ if there exists a QPT algorithm $\mathcal{Q}$ and a polynomial $p$ such that
    for any $x\in\{0,1\}^*$,
    \begin{enumerate}
        \item If $x\in L$, then there exists a $p(\lvert x\rvert)$-bit string ${w_x}$ 
        such that 
        \begin{align}
        \Pr[1\gets\mathcal{Q}(x,{w_x})]\geq \frac{2}{3}.  
        \end{align}
        \item If $x\notin L$, 
        then for any $p(\lvert x\rvert)$-bit string $w$, 
        \begin{align}
        \Pr[1\gets\mathcal{Q}(x,{w})]\leq \frac{1}{3}. 
        \end{align}
    \end{enumerate}
\end{definition}

We also recall the notion of $\textsf{PromiseQCMA}$ problems.

\begin{definition}[$\mathsf{PromiseQCMA}$]\label{def:promise_QCMA}
    A promise problem $(L_Y,L_N)$ is in $\mathsf{PromiseQCMA}$ if there exists a QPT algorithm $\mathcal{Q}$ and a polynomial $p$ such that for any $x\in L_Y\cup L_N$,
    \begin{enumerate}
        \item If $x\in L_Y$, then there exists a $p(\lvert x\rvert)$-bit string ${w_x}$ 
        such that 
        \begin{align}
        \Pr[1\gets\mathcal{Q}(x,{w_x})]\geq \frac{2}{3}.  
        \end{align}
        \item If $x\in L_N$, 
        then for any $p(\lvert x\rvert)$-bit string $w$, 
        \begin{align}
        \Pr[1\gets\mathcal{Q}(x,{w})]\leq \frac{1}{3}. 
        \end{align}
    \end{enumerate}
     \end{definition}


In this work, we sometimes need to provide oracle access to a \textsf{PromiseQCMA} oracle. Notice that defining \textsf{PromiseQCMA} oracle is ambiguous since it is not clear how the oracle should evaluate non-promise inputs, i.e. inputs that do not belong to either $L_Y$ or $L_N$ in a promise problem $(L_Y,L_N)$. Care is needed since the behavior on these inputs may leak extra information that strengthens the power of the oracle. 

To address this issue, we define $\mathcal{O}_{\textsf{PromiseQCMA}}$ to be the set of all possible algorithms that solve $\textsf{PromiseQCMA}$ problems. An algorithm solves $\textsf{PromiseQCMA}$ problems if for any \textsf{PromiseQCMA} problem $(L_Y,L_N)$ with corresponding verification algorithm $\mathcal{Q}$, the algorithm sends the input $(\mathcal{Q},x)$ with $x\in L_Y$ to 1 and $(\mathcal{Q},x)$ with $x\in L_N$ to 0. However, there is no restriction on the behavior on non-promise inputs. 

In the case $\adv^{c\QCMA}$, with $\QCMA\in \mathcal{O}_{\textsf{PromiseQCMA}}$, performs a task such as distinguishing between two distributions, indistinguishability requires that  $\adv^{c\QCMA}$ cannot distinguish for at least a single oracle $\QCMA\in \mathcal{O}_{\textsf{PromiseQCMA}}$. In this way, $\adv$ is a ``valid'' distinguisher only if it can use the oracle's power in solving promise inputs to distinguish rather than its behavior on non-promise inputs (see \cref{sec:subset}). It may be unclear how to define oracle access to a quantum algorithm, however, in our work, we only consider \emph{classical} oracle access $\QCMA$, so this is not an issue. Note that allowing a form of quantum access would strengthen the assumption and weaken our impossibility result. 

\begin{definition}[\textsf{PromiseQCMA} oracles]
The  set $\mathcal{O}_{\textsf{PromiseQCMA}}$ consists of any (computationally-unbounded) algorithm $O$ that satisfies the following condition: 

For any QPT algorithm $\mathcal{Q}$ and for some polynomial $p$ dependent on $\mathcal{Q}$, for any $x\in \{0,1\}^*$, 
\begin{align}
 O(\mathcal{Q},x )=
 \begin{cases}
 1 & \text{if  } \ \exists w_x \in \{0,1\}^{p(\lvert x\rvert )} : \ \Pr[1\gets\mathcal{Q}(x,{w}_x)]\geq 2/3\\
       0 & \text{if  }\ \forall w \in \{0,1\}^{p(|x|)} :\  \Pr[1\gets\mathcal{Q}(x,{w})]\leq 1/3.
 \end{cases}
\end{align}  
\end{definition}

We will also require $\textsf{PromiseQCMA}$ oracle access relative to another oracle $\mathcal{T}$, which we define as follows. 

\begin{definition}[\textsf{PromiseQCMA} oracles relative to another oracle]
For any quantum unitary oracle $\mathcal{T}$, the  set $\mathcal{O}_{\textsf{PromiseQCMA}}^\mathcal{T}$ consists of any (computationally-unbounded) algorithm $O$ that satisfies the following condition: 

For any QPT oracle-aided algorithm $\mathcal{Q}^\mathcal{T}$ and  for some polynomial $p$ dependent on $\mathcal{Q}$, 
for any $x\in \{0,1\}^{*}$, 
\begin{align}
 O(\mathcal{Q},x )=
 \begin{cases}
 1 & \text{if  } \ \exists w_x \in \{0,1\}^{p(\lvert x\rvert )}: \ \Pr[1\gets\mathcal{Q}^\mathcal{T}(x,{w})]\geq 2/3\\
       0 & \text{if  }\ \forall w \in \{0,1\}^{p(|x|)}: \ \Pr[1\gets\mathcal{Q}^\mathcal{T}(x,{w})]\leq 1/3.
 \end{cases}
\end{align}  
\end{definition}



\subsection{Indistinguishable Distributions}
\label{sec:oracle-indis}

We define notions of indistinguishability between distributions. 

\begin{definition}
Let $s$ and $\epsilon $ be functions on the security parameter $n\in \mathbb{N}$. 
We say that two distributions $\textbf{D}^0$ and $\textbf{D}^1$ 
are \emph{$(s,\epsilon)$-sample-indistinguishable} if for every $s$-time quantum algorithm $\D$,
\begin{align}
    \left| \Pr_{x\leftarrow \Dis_n^0}\left[\D(x)=1\right]-\Pr_{x\leftarrow \Dis_n^1}\left[\D ( x)=1\right]\right|\leq \epsilon(n)
\end{align}
for all sufficiently large $n\in\mathbb{N}$.
\end{definition}

We also define the notion of quantum-oracle-indistinguishability for two tuples of distributions. See \cref{sec:compressed} for how we define sampling a function $O$ from a tuple of distributions $\mathcal{D}$.

\begin{definition}[Quantum-Oracle-Indistinguishability]
\label{def:quantum-oracle-indis}
    Let $k,s, q, \epsilon$ be functions on the security parameter $n$. Let $\Dis^0\coloneqq (\Dis^0_x)_{x\in [k]}$ and  $\Dis^1\coloneqq (\Dis^1_x)_{x\in [k]}$ be two tuples of distributions.
 We say that $\Dis^0$ and $\Dis^1$ are $(s,q,\epsilon)$-\emph{quantum-oracle-indistinguishable} if for every $q$-query, $s$-time quantum algorithm $\D$, 
 \begin{align}
    \left| \Pr_{O\gets \Dis^0}[\D^{q(n)O}(1^n)=1]-\Pr_{O\gets \Dis^1}[\D^{q(n)O}(1^n)=1]\right|< \epsilon(n)
\end{align}
{for all sufficiently large $n\in\mathbb{N}$.} 
\end{definition}

We recall how to upgrade sample-indistinguishability to oracle-indistinguishability from \cite{ADS+24}.

\begin{lemma}[Theorem 17 in \cite{ADS+24}]
\label{lem:oracle indis}
 Let $k,p,q',s,\epsilon$ be functions on the security parameter $n\in \mathbb{N}$.  
 Let $\Dis^0\coloneqq (\Dis^0_x)_{x\in [k]}$ and 
 $\Dis^1\coloneqq (\Dis^1_x)_{x\in [k]}$ 
 be two tuples of distributions on a set $Y$ consisting of strings of length at most $p(n)$.
 
 Assume tha for any $x\in [k]$,\mor{If $k$ is a function of $n$, which $n$? $\textbf{D}^0_x$ and $\textbf{D}^1_x$ are sample indistinguishable means $n$ is scaled, so for any $x\in[k]$ for a fixed $n$ is strange.}\mo{I am not sure I understand, but this just means that no adv can distinguish these two distributions with better than $\epsilon(n)$ probability.}
 \takashi{I believe Tomoyuki means the following: When we say "$\textbf{D}^0_x$ and $\textbf{D}^1_x$ are $(s,\epsilon)$-sample-indistinguishable", $\textbf{D}^0_x$ and $\textbf{D}^1_x$ should be treated as a sequence of distributions parameterized by $n$. But before that, we fix one index $x\in[k]$, and then it's unclear how $\textbf{D}^0_x$ and $\textbf{D}^1_x$ are defined as a sequence of distributions. An answer to this question would be that, "$x\in[k]$" actually means "$x_n\in[k(n)]$ for all $n\in\mathbb{N}$", and similarly $\textbf{D}^0_x$ and $\textbf{D}^1_x$ actually mean $\{\textbf{D}^0_{x_n}\}_n$ and $\{\textbf{D}^1_{x_n}\}_n$. This is surely confusing notation, but I believe this is widely accepted convention in cryptography.
 }\mor{Yes, that is my concern.}
 the distributions $\textbf{D}^0_{x}$ and $\textbf{D}^1_x$ are $(s,\epsilon)$-sample-indistinguishable. Then, $\Dis^0$ and $\Dis^1$ are \emph{$(s',q',\epsilon')$-quantum-oracle-indistinguishable}, for any functions $s'$ and $\epsilon'$ satisfying:
     \begin{align}
         s(n)&\geq s'(n)+\frac{q'(n)^3k(n)^3p(n)}{\epsilon(n)},\\
         \epsilon(n)&\leq O\left(\frac{\epsilon'(n)^2}{k(n)^2q'(n)^3}\right).
     \end{align}
     for all $n\in \mathbb{N}$. 
\end{lemma}

We will also use the following result known as the Borel-Cantelli Lemma \cite{B909,C917,S16}. 

\begin{lemma}[Borel-Cantelli Lemma]
\label{lem:BC}
    If the sum of probabilities of events $\{E_n\}_{n\in \mathbb{N}}$ is finite, 
    i.e., $\sum_{n=1}^\infty \Pr[E_n]<\infty$, then the probability that infinitely many of these events occur is 0.
\end{lemma}

\subsection{Non-Interactive Classical Verification of Quantum Computation}

The goal of a \emph{non-interactive classical verification of quantum computation} ($\NICV$) is for a quantum prover to convince a verifier of the validity of a certain quantum computation (such as membership of an element in a $\QMA$ language) non-interactively with classical communication.  Traditionally, the verifier is classical since the goal is to enable classical users to delegate quantum computations. However, we note that our black-box separation also applies to quantum verifiers, as long as the communication is classical. Since we only study {non-interactive} protocols, our definition is restricted to this case.   

\begin{definition}[$\NICV$]
\label{def:NICV}
    Let 
    $L\subseteq \{0,1\}^*$ be a language in $\QMA$ with a verification algorithm $\mathcal{Q}$ and a relation $\mathcal{R}_{L,\mathcal{Q}}$.
    A \emph{non-interactive classical verification of quantum computation} ($\NICV$) scheme
    for $L$ and $\mathcal{Q}$
    is a set $(\mathcal{G},\mathcal{P},\mathcal{V})$ of QPT algorithms that satisfy the following properties:
    \begin{itemize}
        \item \emph{Completeness:} There exists a polynomial $t=t(n)$ such that for all $(x,\ket{w_x})\in \mathcal{R}_{L,\mathcal{Q}}$, 
        \begin{align}
            \Pr\left[ \begin{tabular}{c|c}
 \multirow{2}{*}{$1\gets\mathcal{V}(\sk, x,\pi) $} &  $\ (\pk,\sk)\ \leftarrow \mathcal{G}(1^n)$ \\ 
 & $\ \pi\ \leftarrow \mathcal{P}\left(\pk,x,\ket{w_x}^{\otimes t(n)}\right)$\\
 \end{tabular}\right] \geq 1-\negl[n].
        \end{align}

        \item \emph{Soundness:} For any QPT algorithm $\overline{\mathcal{P}}$,  \begin{align}
            \Pr\left[ \begin{tabular}{c|c}
 \multirow{2}{*}{$1\gets\mathcal{V}(\sk, \overline{x},\overline{\pi}) \wedge \overline{x}\notin L \ $} &  $\ (\pk,\sk)\ \leftarrow \mathcal{G}(1^n)$ \\ 
 & $\ (\overline{x},\overline{\pi})\ \leftarrow \overline{\mathcal{P}}(1^n,\pk)$\\
 \end{tabular}\right] \leq \negl[n].
        \end{align}
    \end{itemize}
\end{definition}

In fact, our impossibility result also excludes $\NICV$ schemes with a more relaxed notion of soundness, where an adversary succeeds with at most a constant $p<1$ probability.

\subsection{Falsifiable Assumptions}

We recall the notion of falsifiable assumptions \cite{N03,GW11,ADS+24}. However, unlike the works \cite{N03,GW11,ADS+24}, we allow for a \emph{quantum} challenger to consider a more general definition.


\begin{definition}[Falsifiable Cryptographic Assumptions]
A \emph{falsifiable cryptographic assumption} 
is a pair
$(\mathcal{C},c)$ 
of a QPT algorithm $\mathcal{C}$ (challenger)
and a constant $c \in [0, 1)$. 
$\mathcal{C}$ interacts with an adversary over a quantum channel and then outputs $1/0$.
The assumption is said to be true if, for any QPT algorithm $\adv$,
\begin{align}
    \Pr\left[ 1\gets  \langle \mathcal{C}(1^n),\adv(1^n)\rangle \right]\leq c+\negl[n],
\end{align}
where $1\gets\langle\mathcal{C}(1^n),\adv(1^n)\rangle$ means that $\mathcal{C}$ outputs $1$ after the interaction
with $\adv$.

\end{definition}


\subsection{Quantum Black-Box Reductions}

Quantum black-box reductions are explored in several works \cite{LP24,TZ25,CM24,CCS24}. We formulate quantum black-box reductions specifically showing the soundness of a $\NICV$ based on a falsifiable cryptographic assumption. Our definition is similar to the definition of black-box reduction for proofs of quantumness in \cite{TZ25}. 


\begin{definition}
Let $L\subseteq\{0,1\}^*$ be a language in $\QMA$ {with a verification algorithm $\mathcal{Q}$}.
Let $\Pi=(\mathcal{G},\mathcal{P},\mathcal{V})$ be a $\NICV$ for $(L,\mathcal{Q})$.
    We say that a classical deterministic algorithm  $\overline{\mathcal{P}}$ is a \emph{$\Pi$-adversary} \footnote{We can generalize this definition by allowing quantum $\Pi$-adversaries, but this can only weaken our impossibility result. } if there exists a polynomial $p$ such that
    \begin{align}
            \Pr\left[ \begin{tabular}{c|c}
 \multirow{2}{*}{$1\gets\mathcal{V}(\sk, \overline{x},\overline{\pi}) \wedge \overline{x}\notin L \ $} &  $\ (\pk,\sk)\ \leftarrow \mathcal{G}(1^n)$ \\ 
 & $\ (\overline{x},\overline{\pi})\ \leftarrow \overline{\mathcal{P}}(1^n,\pk)$\\
 \end{tabular}\right] \geq \frac{1}{p(n)}
        \end{align}
    for infinitely many $n\in \mathbb{N}$. In this case, we say $\overline{\Pv}$ has advantage $1/p$. 
    
    Similarly, we say that a set of algorithms $\overline{\textbf{P}}$  is a $\Pi$-adversary if there exists a polynomial $p$ such that
    \begin{align}
            \Pr_{\overline{\Pv}\gets \overline{\textbf{P}}}\left[ \begin{tabular}{c|c}
 \multirow{2}{*}{$1\gets\mathcal{V}(\sk, \overline{x},\overline{\pi}) \wedge \overline{x}\notin L \ $} &  $\ (\pk,\sk)\ \leftarrow \mathcal{G}(1^n)$ \\ 
 & $\ (\overline{x},\overline{\pi})\ \leftarrow \overline{\mathcal{P}}(1^n,\pk)$\\
 \end{tabular}\right] \geq \frac{1}{p(n)}
        \end{align}
    for infinitely many $n\in \mathbb{N}$. 
\end{definition}

\begin{definition}
\label{def:bb reduction}
Let $\Pi=(\mathcal{G},\mathcal{P},\mathcal{V})$ be a $\NICV$. 
    A \emph{quantum black-box reduction} showing the soundness of $\Pi$ from a falsifiable cryptographic assumption $(\mathcal{C},c)$ is a QPT algorithm 
    $\Sigma^{(\cdot)}$ such that for any polynomial $p$, there exists a polynomial $p'$ such that for any (even inefficient) $\Pi$-adversary $\overline{\Pv}$ with advantage $1/p$, $\Sigma^{{\overline{\Pv}}}$ breaks the assumption with advantage $1/p'$ i.e.
    \begin{align}
    \Pr[1\gets\langle \mathcal{C}(1^n),\Sigma^{{\overline{\Pv}}}(1^n)\rangle]\geq c+\frac{1}{p'(n)},
\end{align}
for infinitely  many $n\in \mathbb{N}$.
\end{definition}



\section{$\QMA \text{-} \QCMA$ Gap Problems Definition}
\label{sec:subset}

We define the notion of \emph{$\QMA \text{-} \QCMA$ gap problems}. We will assume the existence of these problems in our separation (\cref{thm:main}). We support the existence of these problems by providing a construction relative to an oracle (\cref{result 2}). 

\begin{definition}[$\QMA \text{-} \QCMA$ gap problems]
\label{def:subset_membership_problem}    
    Let $s:\mathbb{N}\rightarrow \mathbb{N}$ and $\epsilon:\mathbb{N}\rightarrow [0,1]$ be functions and $t$ be a positive polynomial. A $(t,s,\epsilon)$-$\QMA \text{-} \QCMA$ gap problem of a language $L$ in $\QMA$ with verification algorithm $\mathcal{Q}$ and a $\QMA$ relation $\mathcal{R}_{L,\mathcal{Q}}$ is a pair $(\mathcal{L},\overline{\mathcal{L}})$ satisfying the following conditions:
    \begin{enumerate}
    \item ${\mathcal{L}}\coloneqq\{\mathcal{L}_n\}_{n\in \mathbb{N}}$ and  $\overline{\mathcal{L}}\coloneqq \{\overline{\mathcal{L}}_n\}_{n\in \mathbb{N}}$ are distribution ensembles, where $\mathcal{L}_n$ is a 
     distribution over $L_n\coloneqq L\cap\{0,1\}^n$ and $\overline{\mathcal{L}}_n$ is a  distribution over $\overline{L}_n\coloneqq \{0,1\}^n\setminus L_n$ for each $n\in \mathbb{N}$.

        \item There exists a QPT algorithm $\textsf{SampYes}$ that takes as input $1^n$ and outputs $(x,\ket{w_x}^{\otimes t(n)})$, where $(x,\ket{w_x})\in \mathcal{R}_{L,\mathcal{Q}}$, such that the projection to the first coordinate of the output is $\mathcal{L}_n$.

        \item There exists a not-necessarily-efficient algorithm $\textsf{SampNo}$ that takes as input $1^n$ and samples from the distribution $\overline{\mathcal{L}}_n$.

    \item For any $s(n)$-time distinguisher $\adv$, there exists an oracle $\QCMA \in \mathcal{O}_{\textsf{PromiseQCMA}}$, such that
        \begin{align}
            \left| \Pr_{x\leftarrow {\mathcal{L}}_n}[\adv^{c\QCMA}(x)=1]-\Pr_{x\leftarrow \overline{\mathcal{L}}_n}[\adv^{c\QCMA}(x)=1]\right| \leq \epsilon(n)
        \end{align}
        for all sufficiently large $n\in \mathbb{N}$. Recall, $c\QCMA$ denotes classical oracle access to $\QCMA$. 
        We say that the problem is \emph{subexponentially hard} if there exists a constant $\delta>0$, such that $s(n)=2^{O(n^\delta)}$ and $\epsilon(n)=1/2^{\Omega(n^\delta)}$.
            \end{enumerate}
\end{definition}

We will support the existence of a $\QMA \text{-} \QCMA$ gap problem in \cref{sec:separating qma and qcma} by giving a construction relative to a quantum unitary oracle. To do this, we define what it means to show existence relative to an oracle. 

\begin{definition}[$\QMA \text{-} \QCMA$ gap problems Relative to Oracles]
    Let $s:\mathbb{N}\rightarrow \mathbb{N}$ and $\epsilon:\mathbb{N}\rightarrow [0,1]$ be functions and $t$ be a polynomial. 
    Let $\mathcal{T}$ be a quantum unitary oracle sampled from some distribution of unitaries. 
    A $(t,s,\epsilon)$-$\QMA \text{-} \QCMA$ gap problem relative to $\mathcal{T}$ of a language $L\in \QMA^\mathcal{T}$ with verification algorithm $\mathcal{Q}^{\mathcal{T}}$,
    and a $\QMA$ relation $\mathcal{R}_{L,\mathcal{Q}^{\mathcal{T}}}$ is a pair $(\mathcal{L},\overline{\mathcal{L}})$ satisfying the following conditions:
    \begin{enumerate}
        \item ${\mathcal{L}}\coloneqq\{\mathcal{L}_n\}_{n\in \mathbb{N}}$ and  $\overline{\mathcal{L}}\coloneqq \{\overline{\mathcal{L}}_n\}_{n\in \mathbb{N}}$ are distribution ensembles, where $\mathcal{L}_n$ is a distribution over $L_n\coloneqq L\cap\{0,1\}^n$ and $\overline{\mathcal{L}}_n$ is a distribution over $\overline{L}_n\coloneqq \{0,1\}^n\setminus L_n$ for each $n\in \mathbb{N}$.

        \item There exists a QPT oracle-access algorithm $\textsf{SampYes}^\mathcal{T}(1^n)$ that outputs $(x,\ket{w_x}^{\otimes t(n)})$ where $(x,\ket{w_x})\in \mathcal{R}_{L,\mathcal{Q}^{\mathcal{T}}}$, such that the projection  to the first coordinate of the output is $\mathcal{L}_n$.

        \item There exists a not-necessarily-efficient \footnote{Our construction of a $\QMA \text{-} \QCMA$ gap problem (\cref{con:oracles}) actually has an efficient sampling procedure for no instances, however, this is not required for the main result.} oracle-access algorithm $\textsf{SampNo}^\mathcal{T}$ that takes as input $1^n$ and samples from the distribution $\overline{\mathcal{L}}_n$. 

        \item For any $s(n)$-time quantum algorithm $\adv$, there exists an oracle $\QCMA \in \mathcal{O}_{\textsf{PromiseQCMA}}^\mathcal{T}$ such that with probability 1 over the choice of oracle $\mathcal{T}$,
        \begin{align}
            \left| \Pr_{x\leftarrow {\mathcal{L}}_n}[\adv^{c\QCMA,\mathcal{T}}(x)=1]
            -\Pr_{x\leftarrow \overline{\mathcal{L}}_n}[\adv^{c\QCMA,\mathcal{T}}( x)=1]\right| \leq \epsilon(n)
        \end{align}
        for all sufficiently large $n\in \mathbb{N}$.
        
    \end{enumerate}
\end{definition}

\section{Indistinguishability with Classical Auxiliary Information}
\label{sec:ind with classical info}

In this section, we show that for any $\QMA$ language $L$ with a $\QMA \text{-} \QCMA$ gap problem, any polynomially-bounded classical auxiliary information does not break the indistinguishability of the language.  


\begin{theorem}
\label{result 1}
    Let $p$ be a polynomial and let $L$ be a language with a $(t,s,\epsilon)$-$\QMA \text{-} \QCMA$ gap problem classified by the pair of distribution ensembles $(\mathcal{L},\overline{\mathcal{L}})$ with ${\mathcal{L}}\coloneqq\{\mathcal{L}_n\}_{n\in \mathbb{N}}$ and  $\overline{\mathcal{L}}\coloneqq \{\overline{\mathcal{L}}_n\}_{n\in \mathbb{N}}$. Assume {for each $n\in\mathbb{N}$,} there exists a distribution $\mathcal{L}_n^*$ on pairs $(x,\pi)$ where $x$ is distributed according to $\mathcal{L}_n$ and $\pi$ is some arbitrary string of length at most $p(n)$.   Then there exists a distribution $\overline{\Lg}_n^*$ on pairs $(\overline{x},\overline{\pi})$ where $\overline{x}$ is distributed according to $\overline{\Lg}_n$ and $\overline{\pi}$ has length at most $p(n)$ such that 
    $\Lg^*\coloneqq\{\Lg^*_n\}_{n\in\mathbb{N}}$ and $\overline{\Lg}^*\coloneqq\{\overline{\Lg}^*_n\}_{n\in\mathbb{N}}$ are $(s^*, \epsilon^*)$-sample-indistinguishable as long as $\epsilon^*(n)\leq 8\epsilon(n)$ and $s(n)\geq {n\cdot s^*(n)\cdot q(n)^3}$ for all $n\in \mathbb{N}$, where $q(n)\coloneqq\frac{1000}{\epsilon^*(n)^2}\left(p(n)+\ln\left(\frac{16}{\epsilon^*(n)}\right)\right)$.  
\end{theorem}

\begin{proof}
Assume for contradiction that the lemma is false. Then, there does not exist a distribution ensemble $\overline{\Lg}^*$ that is $(s^*,\epsilon^*)$-sample-indistinguishable from $\Lg^*$.

Let $\textsf{time}(m)$ be the set of quantum algorithms with run-time {at most} $m$ and let $\textsf{dist}(m)$ be the set of distributions over the set $\textsf{time}(m)$. Let $\textsf{Dist}(\overline{\mathcal{L}}_n)$ be the set of joint distributions on pairs $(\overline{x},\overline{\pi})$ with component $\overline{x}$ distributed according to $\overline{\mathcal{L}}_n$ and second component is of length at most $p(n)$.

{For infinitely many $n\in\mathbb{N}$,} we have the following bound
\begin{align}
     \epsilon^*(n)&\leq \underset{\overline{\Lg}^*_n\in \textsf{Dist}(\overline{\Lg}_n)}{\min} \underset{\mathcal{D}\in \textsf{time}(s^*(n))}{\max}\left|\Pr_{(\overline{x},\overline{\pi})\leftarrow \overline{\Lg}_n^*}[\mathcal{D}(\overline{x},\overline{\pi})=1]-\Pr_{({x},{\pi})\leftarrow {\Lg}_n^*}[\mathcal{D}({x},{\pi})=1]\right|\\
     \label{eq:o1}
     &\leq \underset{\overline{\Lg}^*_n\in \textsf{Dist}(\overline{\Lg}_n)}{\min} \underset{\mathcal{D}\in \textsf{time}(s^*(n)+1)}{\max}
     \left[\Pr_{(\overline{x},\overline{\pi})\leftarrow \overline{\Lg}_n^*}[\mathcal{D}(\overline{x},\overline{\pi})=1]-\Pr_{({x},{\pi})\leftarrow {\Lg}_n^*}[\mathcal{D}({x},{\pi})=1]\right]\\
     \label{eq:41}
        &= \underset{\overline{\Lg}^*_n\in \textsf{Dist}(\overline{\Lg}_n)}{\min} \underset{\mathcal{D}\in \textsf{time}(s^*(n)+1)}{\max}\underset{(\overline{x},\overline{\pi})\leftarrow \overline{\Lg}_n^*}{\mathbb{E}}\left[\mathcal{D}(\overline{x},\overline{\pi})-\Pr_{({x},{\pi})\leftarrow {\Lg}_n^*}[\mathcal{D}({x},{\pi})=1]\right]\\
                &\leq \underset{\overline{\Lg}^*_n\in \textsf{Dist}(\overline{\Lg}_n)}{\min} \underset{\Dis_n\in \textsf{dist}(s^*(n)+1)}{\max}\underset{\begin{subarray}{c}(\overline{x},\overline{\pi})\leftarrow \overline{\Lg}_n^*\\ \mathcal{D}\leftarrow \Dis_n \end{subarray}}{\mathbb{E}}\left[\mathcal{D}(\overline{x},\overline{\pi})-\Pr_{({x},{\pi})\leftarrow {\Lg}_n^*}[\mathcal{D}({x},{\pi})=1]\right]\\ \label{eq:o2}
                  &=\underset{\Dis_n\in \textsf{dist}(s^*+1)}{\max}\underset{\overline{\Lg}^*_n\in \textsf{Dist}(\overline{\Lg}_n)}{\min} \underset{\begin{subarray}{c}(\overline{x},\overline{\pi})\leftarrow \overline{\Lg}_n^*\\ \mathcal{D}\leftarrow \Dis_n \end{subarray}}{\mathbb{E}}\left[ \mathcal{D}(\overline{x},\overline{\pi})-\Pr_{({x},{\pi})\leftarrow {\Lg}_n^*}[\mathcal{D}({x},{\pi})=1]\right]\\ \label{eq:o3}
     &= \underset{\Dis_n\in \textsf{dist}(s^*+1)}{\max}\underset{\overline{\Lg}^*_n\in \textsf{Dist}(\overline{\Lg}_n)}{\min} 
     \left\{\underset{\begin{subarray}{c}(\overline{x},\overline{\pi})\leftarrow \overline{\Lg}_n^*\\ \mathcal{D}\leftarrow \Dis_n \end{subarray}}{\mathbb{E}}[\mathcal{D}(\overline{x},\overline{\pi})]-\underset{\begin{subarray}{c}({x},{\pi})\leftarrow {\Lg}_n^*\\ \mathcal{D}\leftarrow \Dis_n\end{subarray}}{\mathbb{E}}[\mathcal{D}({x},{\pi})]\right\}.
\end{align}

\cref{eq:o1} follows by possibly negating the output of the distinguisher to make the difference positive. \cref{eq:41} follows since $\underset{(\overline{x},\overline{\pi})\leftarrow \overline{\Lg}_n^*}{\mathbb{E}}[\mathcal{D}(\overline{x},\overline{\pi})]=\linebreak \Pr_{(\overline{x},\overline{\pi})\leftarrow \overline{\Lg}_n^*}[\mathcal{D}(\overline{x},\overline{\pi})=1]$. \cref{eq:o2} follows from the min-max theorem \cite{v28}. 

Let $\Dis_n$ be a distribution that maximizes \cref{eq:o3}, and let $\overline{\Lg}_n^*$ be a corresponding minimizing distribution. Define the following terms
\begin{align}
    \textsf{Val}(x,\pi)&\coloneqq \underset{\mathcal{D}\leftarrow \Dis_n}{\mathbb{E}}[\mathcal{D}(x,\pi)] \\
    \textsf{Val}_{min}(x)&\coloneqq \underset{\pi}{\min}\textsf{Val}(x,\pi)\\
    \overline{\rho}_n&\coloneqq \underset{\overline{x}\leftarrow \overline{\Lg}_n}{\mathbb{E}}[\textsf{Val}_{min}(\overline{x})]
    =\underset{\begin{subarray}{c}(\overline{x},\overline{\pi})\leftarrow \overline{\Lg}_n^*\\ \mathcal{D}\leftarrow \Dis_n\end{subarray}}{\mathbb{E}}[\mathcal{D}(\overline{x},\overline{\pi})]\label{eq:bar_rho}\\
    \rho_n&\coloneqq \underset{{x}\leftarrow {\Lg}_n}{\mathbb{E}}[\textsf{Val}_{min}(x)]\leq \underset{\begin{subarray}{c}({x},{\pi})\leftarrow {\Lg}_n^*\\ \mathcal{D}\leftarrow \Dis_n\end{subarray}}{\mathbb{E}}[\mathcal{D}({x},{\pi})]\label{eq:rho}
\end{align}

By definition of $\overline{\Lg}_n^*$, any pair $(\overline{x},\overline{\pi})$ in the support of $\overline{\Lg}_n^*$ satisfies $\textsf{Val}(\overline{x},\overline{\pi})=\textsf{Val}_{min}(\overline{x})$. Next, by \cref{eq:rho,eq:bar_rho,eq:o3}, we get $    \overline{\rho}_n-\rho_n\geq \epsilon^*$.

We will use the distinguishers $\Dis_n$ to construct a distinguisher $\tilde{\mathcal{D}}^{c\QCMA}$ that breaks indistinguishability of the gap problem for any algorithm $\QCMA \in \mathcal{O}_{\textsf{PromiseQCMA}}$. 

Set $q=q(n)=\frac{1000}{\epsilon^*(n)^2}\left(p(n)+\ln\left(\frac{16}{\epsilon^*(n)}\right)\right)$. For any input $x\in \{0,1\}^n$, the algorithm $\tilde{\mathcal{D}}^{c\PQCMA}(x)$ is defined as follows. 

   \smallskip \noindent\fbox{%
  \parbox{\textwidth}{%
\textbf{Algorithm} $\tilde{\mathcal{D}}^{c\PQCMA}(x):$
{\small
\begin{enumerate}
    \item Sample $q^3$ distinguishers $\mathcal{D}_1,\ldots,\mathcal{D}_{q^3}$ from $\Dis_n$. \textit{(This may not be efficiently samplable, but this choice will later be fixed as advice.)}  Let $\mathcal{D}_{(q^3)}(x,\pi)$ be the algorithm that computes $p_1 \coloneqq \frac{1}{q^3}\sum_{i\in [q^3]}\mathcal{D}_i(x,\pi)$ and returns 1 with probability $p_1$ and 0 otherwise.
\item For any $t\in [0,1]$, define the following algorithm:
\begin{align}
    {V}^{t}_n(y,\pi)=
    \begin{cases}
       1 & \mathcal{D}_{(q^3)}(y,\pi)\leq t \\
       0 & \text{otherwise}.
    \end{cases}
\end{align}

\item For $i=1,\ldots,q:$ 
    \begin{enumerate}
        \item Query the oracle to compute $\PQCMA(V_n^{\frac{i}{q}},x)$.
        \item If the response is 1, set the variable $\rho_x$ to $\frac{i}{q}$ and abort loop.
    \end{enumerate} 
    \item Output 1 with probability $\rho_x$.
\end{enumerate}
}}}

Sample $\mathcal{D}_1,\ldots,\mathcal{D}_{q^3}$ from $\Dis_n$. Define the following terms:
\begin{align}
    t_{x,\pi}&\coloneqq t_{x,\pi}(\D_1,\ldots,\D_{q^3})= \mathbb{E}(\mathcal{D}_{(q^3)}(x,\pi))\\
    t_x&\coloneqq t_x(\D_1,\ldots,\D_{q^3})=\min_{\pi}t_{x,\pi}(\D_1,\ldots,\D_{q^3}).
\end{align}
Let $\pi^*_x$ be a value that achieves this minimum. 

By the Chernoff bound, we have
\begin{align}
    \Pr_{\D_1,\ldots,\D_{q^3}\gets \Dis_n^{\otimes q^3}}\left[\lvert t_{x,\pi}-\textsf{Val}(x,\pi)\rvert \geq \frac{\epsilon^*}{8}\right]\leq 2e^{-q(n)\epsilon^*(n)^2/128}\leq 2^{-p(n)}\frac{\epsilon^*}{8}.
\end{align}

Using a union bound, this gives
\begin{align}
    \Pr_{\D_1,\ldots,\D_{q^3}\gets \Dis_n^{\otimes q^3}}\left[\exists \pi\in \{0,1\}^{p(n)} \text{ s.t. } \lvert t_{x,\pi}-\textsf{Val}(x,\pi)\rvert\geq \frac{\epsilon^*}{8}\right]\leq \frac{\epsilon^*}{8}.
\end{align}
Therefore, 
\begin{align}
\label{eq:1}
     \Pr_{\D_1,\ldots,\D_{q^3}\gets \Dis_n^{\otimes q^3}}\left[\left| t_{x}-\textsf{Val}_{min}(x)\right| \geq \frac{\epsilon^*}{8}\right]\leq \frac{\epsilon^*}{8}.
\end{align}

\begin{lemma}
\label{lem:average}
    Let $x\in \{0,1\}^n$. For any evaluation $\tilde{\mathcal{D}}^{c\PQCMA}(x)$,
    \begin{align}
    \lvert t_x-\rho_x\rvert\leq 2/q.
\end{align}
\end{lemma}

\begin{proof}
We analyze a single evaluation of $\tilde{\mathcal{D}}^{c\PQCMA}(x)$. The first step is to sample distinguishers $\mathcal{D}_1,\ldots,\mathcal{D}_{q^3}$. These distinguishers can now be considered fixed for the rest of the proof.  

For any pair $(x,\pi)$, $\mathcal{D}_{(q^3)}(x,\pi)$ can be viewed as a random variable that is the average of $q^3$ independent random variables $\mathcal{D}_1(x,\pi),\ldots, \mathcal{D}_{q^3}(x,\pi)$. The variance of each of these variables is bounded by 1, so
\begin{align}
    \textsf{Var}\left(\mathcal{D}_{(q^3)}(x,\pi)\right)&=\textsf{Var}\left(\frac{1}{q^3}\sum_{i\in [q^3]}\mathcal{D}_i(x,\pi)\right)
    =\frac{1}{q^6}\textsf{Var}\left(\sum_{i\in [q^3]}\mathcal{D}_i(x,\pi)\right)\\
    &=\frac{1}{q^6}\sum_{i\in [q^3]}\textsf{Var}(\mathcal{D}_i(x,\pi)) \leq \frac{1}{q^3}.
\end{align}
By Chebyshev's inequality, we get
\begin{align}
\label{eq:0}
        \Pr\left[\left| t_{x,\pi}-D_{(q^3)}(x,\pi)\right| \geq \frac{1}{16q}\right]\leq \frac{256}{q}.
\end{align}

There exists an integer $i_x\in [q]$  such that $\frac{i_x-1}{q}\leq t_x<\frac{i_x}{q}$.

Given that $t_{x,\pi^*_x}=t_x$, we have
\begin{align}
\label{eq:2}
    \Pr\left[\mathcal{D}_{(q^3)}(x,\pi^*_x)\leq \frac{(i_x+1)}{q}\right]\geq 1-\frac{256}{q}>2/3.
\end{align}
This means that $\PQCMA(V_n^{\frac{i_x+1}{q}},x)=1$. Furthermore, by the minimality of $\pi^*_x$, for any $\pi\in \{0,1\}^{p(n)}$ and $j<i_x-2$: 
\begin{align}
\label{eq:5}
    \Pr\left[\mathcal{D}_{(q^3)}(x,\pi)\leq \frac{j}{q}\right]<1/3
\end{align}
meaning that $\PQCMA(V_n^{\frac{j}{q}},x)=0$. In other words, \cref{eq:2,eq:5} imply that ${\rho}_x\in \{\frac{i_x-1}{q},\frac{i_x}{q},\frac{i_x+1}{q}\}$, so
\begin{align}
\label{eq:3}
    \lvert t_x-\rho_x\rvert\leq 2/q.
\end{align} 
    \qed
\end{proof}

By \cref{lem:average}, \cref{eq:1}, and the triangle inequality, we get
\begin{align}
         \Pr_{\D_1,\ldots,\D_{q^3}\gets \Dis_n^{\otimes q^3}}&\left[\lvert \rho_{x}-\textsf{Val}_{min}(x)\rvert \geq \frac{\epsilon^*}{8}+\frac{2}{q}\right]\\
         &\leq \Pr_{\D_1,\ldots,\D_{q^3}\gets \Dis_n^{\otimes q^3}}\left[\lvert \rho_{x}-t_x\rvert +\lvert t_x-\textsf{Val}_{min}(x)\rvert \geq \frac{\epsilon^*}{8}+\frac{2}{q}\right]\\ 
         &\leq \Pr_{\D_1,\ldots,\D_{q^3}\gets \Dis_n^{\otimes q^3}}\left[\lvert t_x-\textsf{Val}_{min}(x)\rvert \geq \frac{\epsilon^*}{8}\right]\leq \frac{\epsilon^*}{8}.
\end{align}

Let $E_x$ be the event that $\lvert \rho_{x}-\textsf{Val}_{min}(x)\rvert \leq \frac{\epsilon^*}{8}+\frac{2}{q}$. We get the following bounds:
\begin{align}
    \Pr_{x\leftarrow \mathcal{L}_n}[\tilde{\mathcal{D}}^{c\PQCMA}(x)=1]&=\sum_{x\in L_n}\Pr[\tilde{\mathcal{D}}^{c\PQCMA}(x)=1]\Pr[x\leftarrow \mathcal{L}_n]\\
    &\leq 
\sum_{x\in {L}_n}\left(\Pr[\tilde{\mathcal{D}}^{c\PQCMA}(x)=1|E_x  ]+\frac{\epsilon^*}{8} \right)\Pr[x\leftarrow \mathcal{L}_n]\\
&\leq \sum_{x\in {L}_n}\left(\textsf{Val}_{min}(x)+\frac{\epsilon^*}{4} +\frac{2}{q}\right)\Pr[x\leftarrow \mathcal{L}_n]=\rho_n+\frac{\epsilon^*}{4}+\frac{2}{q}.
\end{align}

\begin{align}
        \Pr_{\overline{x}\leftarrow \overline{\mathcal{L}}_n}[\tilde{\mathcal{D}}^{c\PQCMA}(\overline{x})=1]&=\sum_{\overline{x}\in \overline{L}_n}\Pr[\tilde{\mathcal{D}}^{c\PQCMA}(\overline{x})=1]\Pr[\overline{x}\leftarrow \overline{\mathcal{L}}_n]\\
    &\geq \sum_{\overline{x}\in \overline{L}_n}\left(\Pr[\tilde{\mathcal{D}}^{c\PQCMA}(\overline{x})=1|E_{\overline{x}}  ]-\frac{\epsilon^*}{8} \right)\Pr[\overline{x}\leftarrow \overline{\mathcal{L}}_n]\\
&\geq \sum_{\overline{x}\in \overline{L}_n}\left(\textsf{Val}_{min}(\overline{x})-\frac{\epsilon^*}{4} -\frac{2}{q}\right)\Pr[\overline{x}\leftarrow \overline{\mathcal{L}}_n]=\overline{\rho}_n-\frac{\epsilon^*}{4}-\frac{2}{q}.
\end{align}

So the advantage of $\tilde{\mathcal{D}}^{c\PQCMA}$ is
\begin{align}
    \Pr_{\overline{x}\leftarrow \overline{\mathcal{L}}_n}[\tilde{\mathcal{D}}^{c\PQCMA}(\overline{x})=1]- \Pr_{x\leftarrow \mathcal{L}_n}[\tilde{\mathcal{D}}^{c\PQCMA}(x)=1]\geq \overline{\rho}_n-\rho_n-\frac{\epsilon^*}{2}-\frac{4}{q}> \frac{\epsilon^*}{8}=\epsilon.
\end{align}

Our distinguisher needs to sample from an arbitrarily complex distribution $\Dis_n$ of circuits. However,  we can fix the choice of $\D_1,\ldots, \D_{q^3}$  by maximizing the distinguishing advantage and providing the description of these optimal circuits as advice. Therefore, we obtain a distinguisher with run-time $s^*(n)O(q^3(n))$, which is smaller than $s(n)$ for large enough $n$, with distinguishing advantage $\epsilon$. This contradicts the indistinguishability of the $(t,s,\epsilon)$-$\QMA \text{-} \QCMA$ gap problem of $L$. Therefore, there must exist a distribution $\overline{\Lg}^*$ that is $(s^*,\epsilon^*)$-sample-indistinguishable from $\Lg^*$.
\qed
\end{proof}

\section{Fake Proofs for No Instances}
\label{sec:fake proofs}

In this section, we show that for any {$\QMA$} language with a subexponential $\QMA\text{-}\PQCMA${ gap problem} and a corresponding $\NICV$, there exists a (computationally-unbounded) algorithm $\overline{\Pv}$ that generates ``fake proofs'' that are indistinguishable from honestly generated proofs.

\begin{theorem}
\label{result 3}
Let $q=q(n)$ be a polynomial on the security parameter $n\in \mathbb{N}$. Let $L$ be a $\QMA$ language with a subexponential $\QMA \text{-} \QCMA$ gap problem and let $\Pi=(\G,\Pv,\V)$ be a $\NICV$ for $L$. 
Then, there exists a set of algorithms $\overline{\textbf{P}}$ satisfying the following conditions:
\begin{itemize}
    \item $\overline{\textbf{P}}$ is a $\Pi$-adversary. Specifically: 
    \begin{align}
            \Pr_{\overline{\Pv}\gets \overline{\textbf{P}}}\left[ \begin{tabular}{c|c}
 \multirow{2}{*}{$1\gets\mathcal{V}(\sk, \overline{x},\overline{\pi}) \wedge \overline{x}\notin L \ $} &  $\ (\pk,\sk)\ \leftarrow \mathcal{G}(1^n)$ \\ 
 & $\ (\overline{x},\overline{\pi})\ \leftarrow \overline{\mathcal{P}}(1^n,\pk)$\\
 \end{tabular}\right] \geq 1-\negl[n].
 \label{toran}
        \end{align}
        \item Oracle access to $\overline{\Pv}\gets \overline{\textbf{P}}$ is efficiently simulatable. In particular, there exists an efficient (stateful) algorithm $\Sim_n$ such that for every $q$-query QPT distinguisher $\mathcal{D}$: 
        \begin{align}
             \left| \Pr_{\overline{\Pv}\gets \overline{\textbf{P}}}[\mathcal{D}^{q{\overline{\mathcal{P}}}}(1^n)=1]-\Pr[\mathcal{D}^{q\mathcal{S}_n}(1^n)=1]\right|\leq \negl.
        \end{align} 
        Here, the distinguisher can query on any input length.
\end{itemize}
\end{theorem}

\begin{proof}
The proof is through multiple steps which are done separately. 

\paragraph{Defining the distribution over yes instances and proofs.}

Let $t=t(n)$ be the polynomial in $n$ that represents the required number of witness copies in the $\NICV$ scheme $\Pi$ as described in \cref{def:NICV}. Given all the algorithms of $\Pi$ are QPT, there exists a large enough integer $d>4$ such that the run-time of $(\sk,\textsf{pk})\leftarrow \mathcal{G}(1^n)$, $(x,\ket{w_x}^{\otimes t(n)})\leftarrow \textsf{SampYes}(1^n)$,
and $\pi\leftarrow \mathcal{P}(\pk, x,\ket{w_x}^{\otimes t(n)})$ are all bounded by  $O(n^d)$. 

By our assumption, there exists a subexponential $\QMA \text{-} \QCMA$ gap problem $(\mathcal{L},\overline{\mathcal{L}})$ over $L$. We can amplify this through a standard complexity leveraging argument  to build a $(t^*,s^*,\epsilon^*)$-$\QMA \text{-} \QCMA$ gap problem with $t^*(n)\in \poly[n]$, $s^*(n)=2^{cn^{d+3}}$, and $\epsilon^*(n)=2^{-n^{d+3}}$ and $c$ is some constant that will be chosen later. 

For any $(\sk,\pk)\leftarrow \mathcal{G}(1^n)$, let $\mathcal{L}^*_{n,\pk}$ be the distribution over pairs $(x,\pi)$ generated by running $(x,\ket{w_x}^{\otimes t^*(n)})\leftarrow \textsf{SampYes}(1^n)$, and $\pi\leftarrow \mathcal{P}(\pk,x,\ket{w_x}^{\otimes t^*(n)})$. 

\paragraph{{Constructing} a $\Pi$-adversary.}

We now construct an inefficient $\Pi$-adversary set $\overline{\textbf{P}}$. Note that if a problem is $(t^*,s^*,\epsilon^*)$-$\QMA \text{-} \QCMA$ gap problem, then it is also $(t^*,s^*,\tilde{\epsilon})$-$\QMA \text{-} \QCMA$ gap problem for $\tilde{\epsilon}\geq \epsilon^*$. Therefore, setting $\tilde{\epsilon}(n)=2^{-n^{d+2}}$ and applying \cref{result 1}, for any $(\sk,\pk)\leftarrow \mathcal{G}(1^n)$, there exists a distribution $\overline{\mathcal{L}}^*_{n,\pk}$ consisting of pairs $(\overline{x},\overline{\pi})$, such that ${\mathcal{L}}^*_{n,\pk}$ and $\overline{\mathcal{L}}^*_{n,\pk}$ are $(s',\epsilon')$-sample-indistinguishable, where $\epsilon'(n)=8\tilde{\epsilon}=8\cdot 2^{-n^{d+2}}$ and $s'(n)=\frac{s^*(n)}{q(n)^3\cdot n}= 2^{\Omega(n^{d+2})}$ as long as $c$ is set to be sufficiently large. 

Let $ \overline{\mathcal{L}}^*\coloneqq \{ \overline{\mathcal{L}}^*_{n,\pk}\}_{n,\pk}$. We define the set $\overline{\textbf{P}}$ as the distribution of functions obtained from sampling from $ \overline{\Lg}^*$. Therefore, $\overline{\Pv}\gets \overline{\textbf{P}}$ maps an input $(1^n,\pk)$ to an element from $\overline{\Lg}_{n,\pk}^*$.

Given the indistinguishability between $ \overline{\mathcal{L}}^*$ and $ {\mathcal{L}}^*$, we have: 
\begin{align}
    \Pr_{\overline{\Pv}\gets \overline{\Lg}^*}\left[\mathcal{V}(\sk, \overline{x},\overline{\pi})=1: \begin{matrix}(\pk,\sk)\ \leftarrow \mathcal{G}(1^n) \\ 
 (\overline{x},\overline{\pi})\ \leftarrow \overline{\mathcal{P}}(1^n,\pk)\\
 \end{matrix} \right] 
 &\geq \\ \Pr\left[\mathcal{V}(\sk, {x},{\pi})=1: \begin{matrix}(\pk,\sk)\ \leftarrow \mathcal{G}(1^n) \\ 
 ({x},{\pi})\ \leftarrow \mathcal{L}^*_{n,\pk}\\
 \end{matrix} \right]  -\epsilon'
 &\geq 1-\negl[n].
\end{align}

Therefore, we have demonstrated the first point of the theorem.


\paragraph{Constructing the simulator $\mathcal{S}_n$.}

We will now describe how to construct an efficient (stateful) simulator that, later, we argue is quantum-oracle-indistinguishable to sampling $\overline{\mathcal{P}}\gets \overline{\textbf{P}}$. We define $\mathcal{S}_n$ according to a threshold $m^*(n)=\lfloor \log^{1/d}(n)\rfloor$.

For $m\ge m^*$, define $\Sim_n$ as the stateful QPT algorithm in \cref{thm:compressed} that simulates an oracle $O\gets \{\mathcal{L}^*_{m,\pk}\}_{m\ge m^*,\pk\in \{0,1\}^{m^d}}$ relative to $q$-query adversaries.  


For $m<m^*$, $\mathcal{S}_n$ responds using a table $\mathcal{T}_n$ that is given as non-uniform advice. The table $\mathcal{T}_n$ consists of values generated by  $\overline{\Lg}^*$. Specifically, the table consists of tuples of the form $(m,\pk,\overline{x}_{m,\pk},\overline{\pi}_{m,\pk})$, one for each triple $(m,\pk)$ where $m\in [m^*]$ and $\pk\in \{0,1\}^{m^d}$. 
Noting that $(\overline{x}_{m,\pk},\overline{\pi}_{m,\pk})$ are of length $O(m(n)^d)$, the size of the table can be bounded by:
\begin{align}
    \left| \mathcal{T}_n\right| & =\sum_{m=1}^{m^*} 2^{m^d}\cdot O(m(n)^d)\\
    &\leq m^*(n)\cdot 2^{m^*(n)^d} \cdot O(m^*(n)^d)\\
    &\leq \poly[n].
\end{align}
$\Sim_n$ uses the table to respond to queries as follows. On input $(1^m,{\pk})$ with $m<m^*$, $\mathcal{S}_n$ outputs ${(\overline{x}_{m,\pk},\overline{\pi}_{m,\pk})}$. 


Notice that the circuit run-time is $\lvert \mathcal{S}_n\rvert \leq \poly[n]\cdot \lvert \mathcal{T}_n\rvert\leq \poly[n]$. 


\paragraph{Oracle-Indistinguishability.}

Finally, we show that $\mathcal{S}_n$ and $\overline{\mathcal{P}}\gets \overline{\textbf{P}}$ are quantum-oracle-indistinguishable. 

Set $\delta=\delta(n)\coloneqq 2^{-\log^{\frac{d+0.5}{d}}(n)}$, which is a negligible function. Assume for contradiction that there exists a QPT, $q$-query distinguisher $\D$ such that:
\begin{align}
   \left| \Pr_{\overline{\Pv}\gets \overline{\textbf{P}}}[\mathcal{D}^{q{\overline{\mathcal{P}}}}(1^n)=1]-\Pr[\mathcal{D}^{q\mathcal{S}_n}(1^n)=1]\right| >\delta
\end{align}
for infinitely many $n\in \mathbb{N}$.

Since $\mathcal{D}$ is QPT algorithm there exists a max security parameter $m_{\textsf{max}}(n)\in \poly[n]$ such that any query by $\D$ is of length at most $m_{\max}(n)$. 

We now consider a set of hybrids $\Hy_{1},\ldots, \Hy_{m_{\max}(n)}$, where $\Hy_j$ is defined as follows:

\smallskip \noindent\fbox{%
  \parbox{0.95\textwidth}{%
$\Hy_j$:
\item Define a stateful algorithm $O_j$ as follows. For every $m\in \{1,\ldots, m_\textsf{max}(n)\}$ and $\pk\in \{0,1\}^{m^d}$:
    \begin{itemize}    
        \item If $m<j$, on input $(1^m,\pk)$, $O_j$ runs  $\mathcal{S}_n(1^m,\pk)$. 
        \item If $m\geq j$, sample $(x_{m,\pk},\pi_{m,\pk})\gets \overline{\mathcal{L}}^*_{m,\pk}$ and set  $O_j(1^m,\pk)\coloneqq (x_{m,\pk},\pi_{m,\pk})$. 
    \end{itemize}
    \item Run $\D^{q{O}_j}(1^n)$ and output the result.}}
\smallskip

\begin{lemma}
\label{lem:small-hybrids}
    For $m<m^*$, hybrids $\Hy_m$ and $\Hy_{m+1}$ are indistinguishable. 
\end{lemma}

\begin{proof}
  This is clear because the output of $\Sim_n$ on any input $(1^m,\pk)$ with $m<m^*$ is sampled from $\overline{\mathcal{L}}^*_{m,\pk}$. 
    \qed
\end{proof}

Now we analyze the distinguishing advantage among the rest of the hybrids. For any $j\in [m_{\max}(n)]$, let
\begin{align}
    \epsilon_j\coloneqq \left| \Pr[\Hy_{j}=1]-\Pr[\Hy_{j-1}=1] \right|.
\end{align}
Set  $j^*=\textsf{arg\ max}_j\epsilon_j$. Note that $\Hy_1$ has the same output distribution as $\mathcal{D}^{{\overline{\mathcal{P}}}}(1^n)$ over the distribution $\overline{\Pv}\gets \overline{\textbf{P}}$, and $\Hy_{m_{\max}}$ has the same output distribution as $\mathcal{D}^{\mathcal{S}_n}(1^n)$. Therefore, we have
\begin{align}
    \delta&<\left| \Pr[\Hy_{m_{\max}(n)}=1]-\Pr[\Hy_{1}=1] \right|\\
         &\leq \sum_{j=2}^{m_{\max}(n)}\left| \Pr[\Hy_{j}=1]-\Pr[\Hy_{j-1}=1] \right|\\
        &\leq \sum_{j=m^*(n)+1}^{m_{\max}(n)}\left| \Pr[\Hy_{j}=1]-\Pr[\Hy_{j-1}=1] \right|\label{eq:disappear}\\
        &\leq (m_{\max}(n)-m^*(n)+1)\cdot \epsilon_{j^*}.
\end{align}

Note that \cref{eq:disappear} is obtained because the hybrids $\Hy_j$ and $\Hy_{j+1}$ have the same distribution when $j<m^*(n)$. This gives $\epsilon_{j^*}\geq \frac{\delta}{(m_{\max}(n)-m^*(n)+1)}$. 

For $m^*(n)\leq j\leq m_{\max}(n)$, we define a distinguisher $\D_j^O(1^n)$, given oracle access to some oracle $O$, as follows. $\D_j$ receives as advice a sample $(x_{m,\pk},\pi_{m,\pk})\gets \overline{\mathcal{L}}^*_{m,\pk}$ for each $m<j$ and  $\pk\in \{0,1\}^{m^d}$. Then, it constructs the following algorithm:
\begin{align}
    O_{\D_j}(1^m,\pk)\coloneqq 
    \begin{cases}
        (x_{m,\pk},\pi_{m,\pk}) & \ m<j,\\
        O(\pk) & \ m=j,\\
        \mathcal{S}_n(1^m,\pk) & \ m>j.
    \end{cases}
\end{align}
$\D_j^O(1^n)$ runs $\D^{O_{\D_j}}(1^n)$ and outputs the result. 

All in all, $\D_j^O(1^n)$ simulates hybrid $\Hy_{j+1}$ when $O=\Sim_n(1^j,\cdot)$ and simulates $\Hy_{j}$ when $O\gets \overline{\mathcal{L}}^*_{j}$, where $\overline{\mathcal{L}}^*_{j}\coloneqq \{\overline{\mathcal{L}}^*_{j,\pk}\}_{\pk\in \{0,1\}^{j^d}}$. Therefore, we have 
\begin{align}
     \left| \Pr[\D_{j}^{\Sim_n(1^j,\cdot)}(1^n)=1]-\Pr_{O\gets \overline{\mathcal{L}}^*_{j}}[\D_{j}^{O}(1^n)=1] \right|=
      \left| \Pr[\Hy_{j}=1]-\Pr[\Hy_{j-1}=1] \right|=\epsilon_{j}.
\end{align}
Regarding the size of $\D_j$, taking into account the size of $\D$ and the advice, we get that $\lvert \D_j\rvert \leq 2^{O(j^d)}$. 


Recall that for $m>m^*$ and $\pk\in \{0,1\}^{m^d}$, we have ${\mathcal{L}}^*_{m,\pk}$ and $\overline{\mathcal{L}}^*_{m,\pk}$ are $(s'(m),\epsilon'(m))$-sample-indistinguishable. Notice that $q(n)\in 2^{O(m^d)}$ and the size of the distributions ${\mathcal{L}}^*_{m,\pk},\overline{\mathcal{L}}^*_{m,\pk}$ is bounded by $2^{2m^d}$. Therefore, applying  \cref{lem:oracle indis}, we  obtain $(s(m),q(n),\epsilon(m))$-quantum-oracle-indistinguishability with $s(m)\coloneqq O(2^{m^{d+1}})$ and $\epsilon(m)\coloneqq O(2^{-m^{d+1}})$ between the functions sampled from these distributions. In other words, we have that for any $s(m)$-time, $q(n)$-query quantum algorithm $\adv$,
\begin{align}
    \left| \Pr_{O\gets \overline{\mathcal{L}}^*_{m}}\left[ \adv^{qO}(1^m)=1\right]- \Pr_{O\gets {\mathcal{L}}^*_{m}}\left[ \adv^{qO}(1^m)=1\right]\right| \leq \epsilon(m).
\end{align}

By \cref{thm:compressed} and the definition of $\Sim_n$, for any $q$-query quantum algorithm $\adv$, 
\begin{align}
    \Pr_{O\gets {\mathcal{L}}^*_{m}}\left[\adv^{qO} (1^n)=1\right]=\Pr\left[\adv^{q\Sim_n(1^m,\cdot)}(1^n)=1\right].
\end{align}
Therefore, we have that for any  $s(m)$-time, $q(n)$-query quantum algorithm $\adv$,
\begin{align}
      \left| \Pr_{{O\gets \overline{\mathcal{L}}^*_{m}}}\left[ \adv^{qO}(1^m)=1\right]- \Pr\left[ \adv^{q\Sim_n(1^m,\cdot)}(1^m)=1\right]\right| \leq \epsilon(m).
\end{align}


However, we have constructed an algorithm $\D_{m}$ for some $m\geq m^*(n)$ of size $2^{O(m^d)}$, which is smaller than $s(m)$ for large enough $n$, that distinguishes oracle access to $O\gets \overline{\mathcal{L}}^*_{m}$ and $\Sim_n(1^m,\cdot)$ with advantage satisfying (for large enough $n$): 
\begin{align}
    \epsilon_{m}\geq \frac{\delta}{(m_{\max}(n)-m^*(n)+1)}=\frac{2^{-\log^{\frac{d+0.5}{d}}(n)}}{(m_{\max}(n)-m^*(n)+1)}> \epsilon(m).
\end{align} 
where the last inequality holds because $m\geq m^*(n)=\lfloor \log^{1/d}(n)\rfloor$, so for large enough $n$,
\begin{align}
    \epsilon(m)\leq \epsilon(m^*)< 2^{- \log^{(d+1)/d}(n)}<\frac{2^{-\log^{\frac{d+0.5}{d}}(n)}}{\poly[n]}
\end{align} 
In other words, for some $m\geq m^*$, $\D_{m}$ breaks $(s(m),q(n),\epsilon(m))$-quantum-oracle-indistinguishability, yielding a contradiction. Therefore, we conclude that for any QPT, $q$-query distinguisher $\D$:
\begin{align}
   \left| \Pr_{\overline{\Pv}\gets \overline{\textbf{P}}}[\mathcal{D}^{q{\overline{\mathcal{P}}}}(1^n)=1]-\Pr[\mathcal{D}^{q\mathcal{S}_n}(1^n)=1]\right| <\delta
\end{align}
for large enough $n\in \mathbb{N}$.

\qed
\end{proof}

\section{Main Result}
\label{sec:main result}
We are now ready to show the main result of this work about the impossibility of quantum black-box reduction of $\NICV$ to falsifiable cryptographic assumptions.

\begin{theorem}
\label{thm:main}
Let $L$ be a language in $\QMA$ with the verification algorithm $\mathcal{Q}$.
Assume that $L$ and $\mathcal{Q}$ have a subexponential $\QMA \text{-} \QCMA$ gap problem.
Let $\Pi=(\G,\Pv,\V)$ be a $\NICV$ for $L$ and $\mathcal{Q}$. Then, for any falsifiable cryptographic assumption $(\mathcal{C},c)$, one of the following statements hold:
\begin{enumerate}
    \item The assumption $(\mathcal{C},c)$ is false.
    \item There is no quantum black-box reduction showing the soundness of $\Pi$ from the assumption $(\mathcal{C},c)$. 
\end{enumerate}
\end{theorem}

\begin{proof}
    Since $\Pi$ satisfies the conditions of \cref{result 3}, there exists an (inefficient) $\Pi$-adversary set $\overline{\textbf{P}}$ such that  
     \begin{align}
            \Pr_{\overline{\Pv}\gets \overline{\textbf{P}}}\left[ \begin{tabular}{c|c}
 \multirow{2}{*}{$1\gets\mathcal{V}(\sk, \overline{x},\overline{\pi}) \wedge \overline{x}\notin L \ $} &  $\ (\pk,\sk)\ \leftarrow \mathcal{G}(1^n)$ \\ 
 & $\ (\overline{x},\overline{\pi})\ \leftarrow \overline{\mathcal{P}}(1^n,\pk)$\\
 \end{tabular}\right] \geq 1-\negl[n].
        \end{align}
        
Assume that statement $(2)$ is not true, i.e. there exists a quantum black-box reduction $\Sigma$ that establishes the soundness of $\Pi$ based on $(\mathcal{C},c)$. In other words, there exists a polynomial $p$ such that, 
    \begin{align}
       \Pr_{\overline{\Pv}\gets \overline{\textbf{P}}}[\langle \mathcal{C}(1^n),\Sigma^{{\overline{\mathcal{P}}}}(1^n)\rangle=1]\geq c+\frac{1}{p(n)}.
    \end{align}

    
By \cref{result 3}, there also exists a (stateful) QPT algorithm $\Sim_n$ such that $\overline{\mathcal{P}}\gets \overline{\textbf{P}}$ and $\Sim_n$ are oracle indistinguishable with respect to QPT adversaries. Given that $\langle \mathcal{C},\Sigma^{(\cdot)}\rangle $ is QPT, for large enough $n\in \mathbb{N}$,     
\begin{align}
    \Pr[\langle \mathcal{C}(1^n),\Sigma^{{\Sim_n }}(1^n)\rangle =1]&\geq  \Pr_{\overline{\Pv}\gets \overline{\textbf{P}}}[\langle \mathcal{C}(1^n),\Sigma^{{\overline{\mathcal{P}}}}(1^n)\rangle=1]-\negl \\
    &\geq \left(c+\frac{1}{p(n)}\right)-\negl \geq c+\frac{1}{2p(n)}. 
\end{align}
Therefore, the QPT algorithm $\Sigma^{{\Sim_n}}$ breaks $(\mathcal{C},c)$, so statement $(1)$ is true.
\qed
\end{proof}

\section{Oracle Construction of \textsf{QMA-QCMA}-Gap Problem}
\label{sec:separating qma and qcma}

In this section, we show the existence of a $\QMA \text{-} \QCMA$ gap problem relative to a quantum unitary oracle. First, we state the following useful result from \cite{AK07}. 

\begin{lemma}[Theorem 3.3 in \cite{AK07}]
\label{lem:oracle}
Let $m$ and $q$ be functions on the security parameter $n$. For an $n$-qubit state $\ket{\psi}$, let $U_{{\psi}}$ be the $n$-qubit unitary that maps $U_\psi |\psi\rangle = -|\psi\rangle$, and $U_\psi |\phi\rangle = |\phi\rangle$ whenever $\langle \phi | \psi \rangle = 0$. Let $I$ be the $n$-qubit identity unitary and let $Q$ be a quantum unitary oracle independent of $\ket{\psi}$. Then, for any $q$-query quantum algorithm $\adv$,
\begin{align}
 \underset{\ket{\psi}\gets \mu_n}{\mathbb{E}}\left[\max_{w\in \{0,1\}^{m}}  \left( \left|\Pr[\adv^{Q,U_{\ket{\psi}}}(1^n,w)=1]-\Pr[\adv^{Q,I}(1^n,w)=1]\right|\right)\right]
    \leq O\left( q\sqrt{\frac{m+1}{2^n}} \right).
\end{align}\end{lemma}

\begin{remark}
    Theorem 3.3 in \cite{AK07} is slightly weaker than the version given above as \cref{lem:oracle} for two reasons. Firstly, in the proof of  Theorem 3.3, the witness $w$ is set as the string that maximizes $\adv$'s acceptance probability in the case it is given oracle access to $U_{\ket{\psi}}$. Their proof shows that $\adv$ cannot distinguish between $U_{\ket{\psi}}$ and $I$ with the aid of such a witness. However, in \cref{lem:oracle}, the witness is instead set to the string that maximizes the actual distinguishing probability. Nevertheless, the same proof of Theorem 3.3 carries through for this version of the witness. Secondly, we include another oracle $Q$, which can be included since it is independent of $\ket{\psi}$ and the adversary is treated as a black-box. 
\end{remark}

We will also use the following lemma which may be of independent interest and is proven in \cref{sec:emulation}. Essentially, this lemma shows  that giving $\QCMA$ oracle access can be emulated using a classical witness, even in the presence of other quantum unitary oracles. 

\begin{lemma}
\label{thm:oracle}
    Let $q$ be a function on the security parameter $n$ and let $\epsilon\leq 1$ be a function on $q$ and $n$. Let $Q$ be a quantum unitary oracle.  Let $U\gets \mathcal{U}$ and $U'\gets \mathcal{U}'$ be two $n$-qubit unitaries sampled from some distributions $\mathcal{U}, \mathcal{U}'$ of unitaries. Assume that the following classical witness aided indistinguishability condition is satisfied: For any  $q$-time quantum algorithm $\mathcal{B}$,
\begin{align}
\underset{\begin{subarray}{c} U\gets \mathcal{U}\\
U'\gets \mathcal{U}'\end{subarray}}{\mathbb{E}}\left[ \max_{w\in \{0,1\}^{q}}  \left( \left|\Pr[\mathcal{B}^{Q,U}(1^n,w)=1]-\Pr[\mathcal{B}^{Q,U'}(1^n,w)=1]\right|\right)\right]
    \leq \epsilon(n,q)
\end{align}
for large enough $n$. 
    Then, for any $U,U'$, there exist oracles $\QCMA_U\in \mathcal{O}_{\textsf{PromiseQCMA}}^{Q,U}$ and $\QCMA_{U'}\in \mathcal{O}_{\textsf{PromiseQCMA}}^{Q,U'}$, such that for any  $q$-time quantum algorithm $\adv$,
    \begin{align}
\underset{\begin{subarray}{c} U \gets \mathcal{U}\\
U'\gets \mathcal{U}'\end{subarray}}{\mathbb{E}}\left[   \left|\Pr[\adv^{Q,U,c\QCMA_U}(1^n)=1]-\Pr[\adv^{Q,U',c\QCMA_{U'}}(1^n)=1]\right|\right]\\
    \leq 2q\cdot \epsilon(n,1000\cdot n\cdot q)+4q\cdot O(2^{-n})
\end{align}
    for large enough $n$.
\end{lemma}

We are now ready to prove the existence of a $\QMA \text{-} \QCMA$ gap problem relative to a quantum unitary oracle. 
\begin{theorem}
\label{result 2}
    Let $0<\delta<\frac{1}{16}$ be a constant and $t=t(n)$ be a polynomial on the security parameter $n\in \mathbb{N}$. There exists a language in $\QMA$ and a corresponding $(t,2^{\delta n},2^{-\delta n})$-$\QMA \text{-} \QCMA$ gap problem relative to a quantum unitary oracle. 
\end{theorem}

\begin{proof}
We first define the language for the gap problem. For each $n\in \mathbb{N}$, let $L_n\subseteq\{0,1\}^n$ be a language of size $\lfloor 2^{n/2}\rfloor$ chosen uniformly at random and let $\overline{L}_n\coloneqq\{0,1\}^n\setminus L_n$. We also define $L\coloneqq \cup_n L_n$ and let $L(x)=1$ if $x\in L$ and $L(x)=0$ otherwise.

Now we define the quantum unitary oracles used to construct the gap problem for $L$ as follows. For simplicity, we let $\ket{0}$ denote the state $\ket{0^k}$ of appropriate dimension $k\in \mathbb{N}$. 

\begin{construct}
\label{con:oracles}
For each $n\in \mathbb{N}$ and $x\in \{0,1\}^n$, sample an $n$-qubit state $\ket{\psi_x}\leftarrow \mu_{n}$.
Let $\mathcal{T}\coloneqq(\textsf{GenYes},\textsf{GenNo},U)$, where
$\textsf{GenYes}\coloneqq\{\textsf{GenYes}_n\}_{n\in\mathbb{N}}$,
$\textsf{GenNo}\coloneqq\{\textsf{GenNo}_n\}_{n\in\mathbb{N}}$, 
and
$U\coloneqq\{U_n\}_{n\in\mathbb{N}}$. 
Each of them is defined as follows: 
    \begin{enumerate}
        \item $\textsf{GenYes}_n:$ A unitary on $(t+1)n+1$ qubits defined as follows:
        \begin{align}
        \textsf{GenYes}_n\coloneqq |\phi_n^{\textsf{Yes}}\rangle \langle 0|+|0\rangle\langle\phi_n^{\textsf{Yes}}|+
        I-|\phi_n^{\textsf{Yes}}\rangle\langle\phi_n^{\textsf{Yes}}|-|0\rangle\langle0|,
        \end{align}
        where $\ket{\phi_n^{\textsf{Yes}}}\coloneqq \frac{1}{\sqrt{|L_n|}}\sum_{x\in L_n}\ket{x}\ket{\psi_x}^{\otimes t(n)}\ket{1}$.
        In other words, $\textsf{GenYes}_n$ is a unitary that ``swaps'' $|0\rangle$ and $|\phi_n^{\textsf{Yes}}\rangle$.


        \item $\textsf{GenNo}_n:$ A unitary on $n+1$ qubits defined as follows:
 \begin{align}
        \textsf{GenNo}_n\coloneqq |\phi_n^{\textsf{No}}\rangle \langle0|+|0\rangle\langle\phi_n^{\textsf{No}}|+
        I-|\phi_n^{\textsf{No}}\rangle\langle\phi_n^{\textsf{No}}|-|0\rangle\langle0|,
        \end{align}
        where $\ket{\phi_n^{\textsf{No}}}\coloneqq\frac{1}{\sqrt{|\overline{L}_n|}}\sum_{x\in \overline{L}_n}\ket{x}\ket{1}$.
        In other words, $\textsf{GenNo}_n$ is a unitary that ``swaps'' $|0\rangle$ and $|\phi_n^{\textsf{No}}\rangle$.

\item $U_n$: A unitary on $2n+1$ qubits defined as follows. 
\begin{align}
U_n\coloneqq\sum_{\begin{subarray}{c}
    b\in \{0,1\}\\ x\in\{0,1\}^n
\end{subarray}}|b\rangle\langle b|\otimes |x\rangle\langle x| \otimes\left[(-1)^{L(x)\cdot b}|\psi_x\rangle\langle\psi_x|+(I-|\psi_x\rangle\langle\psi_x|)\right].
\end{align}
In other words, it adds the phase $(-1)^{L(x)\cdot b}$ only to $|\psi_x\rangle$.
    \end{enumerate}    
\end{construct}

We will now show that the oracle $\mathcal{T}$ gives a $\QMA \text{-} \QCMA$ gap problem for $L$. 

Firstly, we need to show that $L\in \QMA^\mathcal{T}$ meaning we need to describe an appropriate verification algorithm $\mathcal{Q}^{\mathcal{T}}$. For $x\in L_n$ with witness $\ket{\psi_x}$, $\mathcal{Q}^{\mathcal{T}}(x,\ket{\psi_x})$ is as follows. Prepare the state 
$\frac{1}{\sqrt{2}}(\ket{0}\ket{x}\ket{\psi_x}+\ket{1}\ket{x}\ket{\psi_x})$, apply $U_n$, apply a Hadamard gate on the first register, measure it, and accept if the result is 1. It is easy to check that this algorithm accepts with probability 1. Next, note that if $\overline{x}\in \overline{L}_n$, then for any $n$-qubit state $\ket{\overline{\psi}}$ and bit $b$, $U_n(\ket{b}\ket{\overline{x}}\ket{\overline{\psi}})=\ket{b}\ket{\overline{x}}\ket{\overline{\psi}}$. Therefore, implementing the same verification procedure on any $\overline{x}\in \overline{L}_n$, given any $n$-qubit state 
as witness, yields 0 with probability 1. Therefore, $\mathcal{Q}^{\mathcal{T}}$ acts as a valid verification algorithm, and we have $L\in \QMA^{\mathcal{T}}$. 

Furthermore, $\textsf{SampYes}$ and $\textsf{SampNo}$ algorithms can be easily constructed using $\textsf{GenYes}$ and $\textsf{GenNo}$, respectively. Specifically, we can define a sampling algorithm $\textsf{SampYes}^\mathcal{T}(1^n)$ that queries $\textsf{GenYes}_n(\ket{0})$, measures the first $n$ qubits of the response in the computational basis, obtains say $x$, collapses the rest of the result to $\ket{x}\ket{\psi_x}^{\otimes t(n)}\ket{1}$, and outputs $(x,\ket{\psi_x}^{\otimes t(n)})$. Similarly, define an algorithm $\textsf{SampNo}^\mathcal{T}(1^n)$ that queries $\textsf{GenNo}_n(\ket{0})$, measures the response in the computational basis, and outputs the the first $n$ bits. 

It remains to show security. Specifically, we will show the indistinguishability condition of the gap problem relative to $\mathcal{T}$. 

Let $\QCMA$ be the algorithm in $\mathcal{O}_{\textsf{PromiseQCMA}}^\mathcal{T}$ that is defined as follows on any non-promise input $(V,x)$. $\QCMA$ receives a witness $w$ that maximizes the probability $\Pr[ V^{\mathcal{T}}(x,w)=1]$, then computes $V^{\mathcal{T}}(x,w)$ $1000\cdot n$ times, and if the average of these computations is less than $1/2$, it outputs 0 and outputs 1 otherwise. Notice that this is the same type of oracle that can be emulated with a classical witness as given in the proof of \cref{thm:oracle}. 
Let $\adv^{c\QCMA,\mathcal{T}}$ be any $2^{\delta n}$-time quantum adversary. Clearly, $\adv$ queries the oracles at most $q\coloneqq 2^{\delta n}$ times.  


We commence with a hybrid argument.

\begin{itemize}
    \item $\Hy_0(n)$:  
    \begin{enumerate}
    \item Sample $L$ and $\mathcal{T}$ as described at the start of the proof.
        \item Sample $y\leftarrow L_n$.
        \item Run $ \adv^{c\QCMA, \mathcal{T}}(y)$, and output the result. 
    \end{enumerate}
    
    \item $\Hy_{1}(n)$: The same as $\Hy_0$, except that $\adv$'s oracle access is modified by replacing the oracle $\textsf{GenYes}_n$ in both the $\QCMA$ and $\mathcal{T}$ oracles (the $\QCMA$ oracle also contains $\textsf{GenYes}_n$) with a new unitary
    $\textsf{GenYes}_n^y$, which is the same as $\textsf{GenYes}_n$ except that $\ket{\phi_n^{\textsf{Yes}}}$ is replaced with the state 
    \begin{align}
        \ket{\phi_n^{\textsf{Yes},y}}\coloneqq \frac{1}{\sqrt{|L_n|-1}}\sum_{x\in L_n\setminus y}\ket{x}\ket{\psi_x}^{\otimes t}\ket{1}.
    \end{align}
    We denote the resulting oracles as $(\QCMA_{1,y},\mathcal{T}_{1,y})$.

        \item $\Hy_{2}(n)$: The same as $\Hy_1$, except that $\adv$'s oracle access is modified by replacing $U_n$ (in both the $\QCMA_{1,y}$ and $\mathcal{T}_{1,y}$ oracles) with a new unitary $U_n^y$, which is defined as follows:
\begin{align}
U_n^y&\coloneqq \sum_{
    b\in \{0,1\}}|b\rangle\langle b|\otimes \left(|y\rangle\langle y|\otimes I    
+ \sum_{x\in \{0,1\}^n\setminus y}|x\rangle\langle x|\otimes \left[(-1)^{L(x)\cdot b}|\psi_x\rangle\langle\psi_x|+(I-|\psi_x\rangle\langle\psi_x|)\right] \right)   .
\end{align}
 We denote the resulting oracles as $(\QCMA_{2,y},\mathcal{T}_{2,y})$.

\item $\Hy_3(n)$: The same as $\Hy_2$, except $\adv$'s oracle access is modified by replacing $\textsf{GenNo}_n$ (in both the $\QCMA_{2,y}$ and $\mathcal{T}_{2,y}$ oracles) with a new unitary $\textsf{GenNo}_n^y$, which is the same as $\textsf{GenNo}_n$ 
except that $\ket{\phi_n^{\textsf{No}}}$ is replaced with the state $\frac{1}{\sqrt{|\overline{L}_n|+1}}\sum_{x\in (\overline{L}_n\cup \{y\})}\ket{x}\ket{1}.$
 We denote the resulting oracles as $(\QCMA_{3,y},\mathcal{T}_{3,y})$. 
            \item $\Hy_4(n)$: We now sample the challenge from $\overline{L}$.
    \begin{enumerate}
    \item Sample $L$ and $\mathcal{T}$ as described at the start of the proof.
        \item Sample $y\leftarrow {L}_n$.
        \item Sample $\overline{y}\gets \overline{L}_n$.
        \item Run $\adv^{c\QCMA_{3,y}, \mathcal{T}_{3,y}}(\overline{y})$, and output the result. 
    \end{enumerate}

       \item $\Hy_{5}(n)$: The same as $\Hy_4$, except that $\adv$'s oracle access is changed back to $(\QCMA_{2,y},\mathcal{T}_{2,y})$.

     \item $\Hy_{6}(n)$: The same as $\Hy_6$, except that $\adv$'s oracle access is changed back to $(\QCMA_{1,y},\mathcal{T}_{1,y})$.

     \item $\Hy_7(n)$: The same as $\Hy_7$, except $\adv$'s oracle access is changed back to $(\QCMA,\mathcal{T})$. In particular, this hybrid is as follows:
         \begin{enumerate}
    \item Sample $L$ and $\mathcal{T}$ as described at the start of the proof.
        \item Sample $\overline{y}\gets \overline{L}_n$.
        \item Run $\adv^{c\QCMA, \mathcal{T}}(\overline{y})$, and output the result. 
    \end{enumerate}
\end{itemize}

\begin{Claim}
\label{claim:1}
\begin{align}
    \left| \Pr[{\Hy}_0(n)=1]-\Pr[\Hy_1(n)=1]\right|\leq O\left(2^{-n/32}\right).
\end{align}
\end{Claim}

\begin{proof}
The only difference between the hybrids is that
$\adv$'s access to $\textsf{GenYes}_n$ is replaced with $\textsf{GenYes}_n^y$.  We first show indistinguishability given a classical witness $w$ instead of a $\QCMA$ oracle. To do this, we run another hybrid argument. 


Fix a witness $w\in \{0,1\}^q$. Let $\Hy_{0,0}(n)$ be the distribution generated by the evaluation $\adv^{\textsf{GenYes},\textsf{GenNo},U}(y,w)$. Let $\Hy_{0,i}(n)$ for $i\in [q]$ be the hybrid where $\adv$'s oracle access is modified as follows: $\adv$ gets oracle access to $(\textsf{GenYes}^y,\textsf{GenNo},U)$ for the first $i$ queries and, then, gets access to $(\textsf{GenYes},\textsf{GenNo},U)$ for the rest of its queries. For any $i\in [q]$, we can bound the ability of $\adv$ to distinguish hybrids $\Hy_{i-1}$ and $\Hy_{i}$ by the spectral norm $\|\cdot \|$ between the unitaries $\textsf{GenYes}_n$ and $\textsf{GenYes}_n^y$. In particular, 
\begin{align}
    &\left|\Hy_{0,i-1}(n)-\Hy_{0,i}(n)\right|=\|\textsf{GenYes}_n- \textsf{GenYes}_n^y\|  \\
   =& \left\|\left(|\phi_n^{\textsf{Yes}}\rangle \langle 0|+|0\rangle\langle\phi_n^{\textsf{Yes}}|-|\phi_n^{\textsf{Yes}}\rangle\langle\phi_n^{\textsf{Yes}}|\right) -\left(|\phi_n^{\textsf{Yes},y}\rangle \langle 0|+|0\rangle\langle\phi_n^{\textsf{Yes},y}|-|\phi_n^{\textsf{Yes},y}\rangle\langle\phi_n^{\textsf{Yes},y}|\right)\right\|\\ \label{eq:n0}
   =&   \left\|\left(|\phi_n^{\textsf{Yes}}\rangle \langle 0|-|\phi_n^{\textsf{Yes},y}\rangle \langle 0|\right) +\left(|0\rangle\langle\phi_n^{\textsf{Yes}}|- |0\rangle\langle\phi_n^{\textsf{Yes},y}| \right)     +\left(|\phi_n^{\textsf{Yes}}\rangle\langle\phi_n^{\textsf{Yes}}|-|\phi_n^{\textsf{Yes},y}\rangle\langle\phi_n^{\textsf{Yes},y}|\right)\right\|\\ \label{eq:n1}
    =&\|  \left(|\phi_n^{\textsf{Yes}}\rangle -|\phi_n^{\textsf{Yes},y}\rangle \right)\langle 0| +|0\rangle \left(\langle\phi_n^{\textsf{Yes}}|- \langle\phi_n^{\textsf{Yes},y}| \right)\\ +& |\phi_n^{\textsf{Yes}}\rangle \left(\langle\phi_n^{\textsf{Yes}}|-\langle\phi_n^{\textsf{Yes},y}|\right)+\left(|\phi_n^{\textsf{Yes}}\rangle-|\phi_n^{\textsf{Yes},y}\rangle\right)\langle\phi_n^{\textsf{Yes},y}| \|\\ \label{eq:n2}
       =&\left\|  |\phi_n^{\textsf{Yes}}\rangle -|\phi_n^{\textsf{Yes},y}\rangle \right\| \cdot \|\langle 0|\| +\| |0\rangle\| \cdot \left\|\langle\phi_n^{\textsf{Yes}}|- \langle\phi_n^{\textsf{Yes},y}| \right\|\\
       +&\| |\phi_n^{\textsf{Yes}}\rangle \|\left\| \langle\phi_n^{\textsf{Yes}}|-\langle\phi_n^{\textsf{Yes},y}|\right\| +\left\| |\phi_n^{\textsf{Yes}}\rangle-|\phi_n^{\textsf{Yes},y}\rangle\right\| \cdot \|\langle\phi_n^{\textsf{Yes},y}| \|\\
    \leq& 2^{-n+3}.
\end{align} 
Here, \cref{eq:n0} is obtained by reordering the terms, \cref{eq:n2} by the triangle inequality, and the final bound by noting:
\begin{align}
    \||\phi_n^{\textsf{Yes}}\rangle -|\phi_n^{\textsf{Yes},y}\rangle\|=\left\| \langle\phi_n^{\textsf{Yes}}|-\langle\phi_n^{\textsf{Yes},y}|\right\|\leq 2^{-n+1}.
\end{align}
This is information-theoretic bound so it holds regardless of the witness or computational power of $\adv$. Therefore, by the triangle inequality, 
\begin{align}
     &\left|\Pr[\adv^{\textsf{GenYes},\textsf{GenNo},U}(y,w)=1]- \Pr[\adv^{\textsf{GenYes}^y,\textsf{GenNo},U}(y,w)=1]\right|  \\ = &\left|\Hy_{0,0}(n)-\Hy_{0,q}(n)\right|\\ \leq & q\cdot 2^{-n+3}.
\end{align}
By \cref{thm:oracle}, we get that for any $q$-time quantum adversary $\adv$,
\begin{align}
     \underset{\ket{\psi_y}\gets \mu_n}{\mathbb{E}} \left[ \left|\Pr[\adv^{c\QCMA,\mathcal{T}}(y)=1]- \Pr[\adv^{c\QCMA_{1,y},\mathcal{T}_{1,y}}(y)=1]\right|  \right]\leq O(2^{-n/32}).
\end{align}
\qed
\end{proof}

\begin{Claim}
\label{claim:2}
\begin{align}
    \left| \Pr[\Hy_1(n)=1]-\Pr[\Hy_2(n)=1]\right|\leq O\left(2^{-n/32}\right).
\end{align}
\end{Claim}

\begin{proof}

The only difference between these hybrids is that $U_n$ is replaced with $U^y_n$, essentially replacing a reflection map with an identity map. As such, distinguishing these two hybrids involves a similar task as \cref{lem:oracle}. Therefore, we have that for any $q$-time quantum algorithm $\adv$,
\begin{align}
 \underset{\ket{\psi}\gets \mu_n}{\mathbb{E}}\left[\max_{w\in \{0,1\}^{q}}  \left( \left|\Pr[\adv^{\mathcal{T}_{1,y}}(y,w)=1]-\Pr[\adv^{\mathcal{T}_{2,y}}(y,w)=1]\right|\right)\right]
    \leq O\left( q\sqrt{\frac{q+1}{2^n}} \right).
\end{align} 
By \cref{thm:oracle}, setting $\epsilon(n,q)\coloneqq O\left( q\sqrt{\frac{q+1}{2^n}} \right)$ as in the equation above, and setting $q=2^{n/16}$ we get 
   \begin{align}
\underset{\ket{\psi}\gets \mu_n}{\mathbb{E}}\left[   \left|\Pr[\adv^{c\QCMA_{1,y},\mathcal{T}_{1,y}}(y)=1]-\Pr[\adv^{c\QCMA_{2,y},\mathcal{T}_{2,y}}(y)=1]\right|\right]
    \leq \\
    2q\cdot \epsilon(n,1000\cdot n\cdot q)+4q\cdot O(2^{-n})\leq O(2^{-\frac{n}{32}})
\end{align}
    for large enough $n$.
    \qed
\end{proof}

\begin{Claim}
\label{claim:3}
\begin{align}
    \left| \Pr[\Hy_2(n)=1]-\Pr[\Hy_3(n)=1]\right|\leq O\left(2^{-n/32}\right).
\end{align}
\end{Claim}

\begin{proof}
This follows in the same way as Claim \ref{claim:1}.
    \qed
\end{proof}

\begin{Claim}
\begin{align}
   \Pr[{\Hy}_3(n)=1]=\Pr[\Hy_4(n)=1].
\end{align}
\end{Claim}

\begin{proof}
Relative to the oracles $(\QCMA_{3,y},\mathcal{T}_{3,y})$, both $y$ and $\overline{y}$ are sampled from the same distribution so they are indistinguishable. 
    \qed
\end{proof}

\begin{Claim}
\begin{align}
    \left| \Pr[{\Hy}_4(n)=1]-\Pr[\Hy_5(n)=1]\right|\leq O\left(2^{-n/32}\right).
\end{align}
\end{Claim}

\begin{proof}
This follows the same way as Claim \ref{claim:1}.
    \qed
\end{proof}

\begin{Claim}
\begin{align}
    \left| \Pr[{\Hy}_5(n)=1]-\Pr[\Hy_6(n)=1]\right|\leq O(2^{- n/32}).
\end{align}
\end{Claim}

\begin{proof}
This follows the same way as Claim \ref{claim:2}.
\qed    
\end{proof}

\begin{Claim}
\begin{align}
    \left| \Pr[{\Hy}_6(n)=1]-\Pr[\Hy_7(n)=1]\right|\leq O\left(2^{-n/32}\right).
\end{align}
\end{Claim}

\begin{proof}
This follows the same way as Claim \ref{claim:1}.
   \qed 
\end{proof}

By the triangle inequality and all the previous claims,
\begin{align}
     \left| \Pr[\Hy_0(n)=1]-\Pr[\Hy_7(n)=1]\right| &\leq \frac{1}{2^{n/100}}.
\end{align}
for large enough $n$. 
However, this probability is taken over the distribution of $L$ and $\mathcal{T}$. By Markov inequality, we get that 
\begin{align}
      \Pr_{{L,\mathcal{T}}} \left[ \left|\Pr_{y\leftarrow {L}_n}[\adv^{c\QCMA, \mathcal{T}}(y)=1]-\Pr_{y\leftarrow \overline{L}_n}[\adv^{c\QCMA, \mathcal{T}}(y)=1]\right| \geq {2^{-n/200}}\right] \leq 2^{-n/200},
    \end{align}
where the first probability is taken over the distribution of $L$ and $\mathcal{T}$. 
By \cref{lem:BC} (Borel-Cantelli Lemma), since $\sum_n 2^{-n/200}$ converges, with probability 1 over the distribution of $L$ and $\mathcal{T}$, it holds that
\begin{align}
     \left| \Pr_{y\leftarrow L_n}[\adv^{c\QCMA, \mathcal{T}}(y)=1]-\Pr_{y\leftarrow \overline{L}_n}[\adv^{c\QCMA, \mathcal{T}}(y)=1]\right|\leq {2^{-n/200}}
    \end{align}
except for finitely many $n\in \mathbb{N}$. There are countable number of quantum algorithms $\adv$ with a run-time and query bound of $2^{n/16}$, so this bound holds for every such adversary.
\qed
\end{proof}

\paragraph{Acknowledgements.}
TM is supported by
JST CREST JPMJCR23I3,
JST Moonshot R\verb|&|D JPMJMS2061-5-1-1, 
JST FOREST, 
MEXT QLEAP, 
the Grant-in Aid for Transformative Research Areas (A) 21H05183,
and 
the Grant-in-Aid for Scientific Research (A) No.22H00522.
Part of this work was done by TM at Columbia University as a visiting scientist.

\printbibliography
\newpage
\appendix

\section{\textsf{QCMA} Oracle Emulation using Classical Witness}
\label{sec:emulation}

In this section, we show that giving $\QCMA$ oracle access can be emulated using a classical witness, even in the presence of other quantum unitary oracles. This is used to bound the distinguishing advantage of an adversary that is given a $\QCMA$ oracle in the construction of a $\QMA \text{-} \QCMA$ gap problem.  
We use the following basic algebraic property.
\begin{lemma}
    \label{lem:bb}
    For any functions $A,B$ mapping $\{0,1\}^n\rightarrow [0,1]$, we have 
    \[
\left|\max_{w\in \{0,1\}^n} A(w)-\max_{w\in \{0,1\}^n} B(w)\right|\le \max_{w\in \{0,1\}^n}|A(w)-B(w)|.
\]
\end{lemma}
\begin{proof}
Let $w_A\coloneqq \argmax_{w\in \{0,1\}^n}A(w)$ and let $w_B\coloneqq \argmax_{w\in \{0,1\}^n}B(w).$ Without loss of generality, assume $|A(w_A)|\geq |B(w_B)|$. Then, 
\begin{align}
 \max_{w\in \{0,1\}^n}|A(w)-B(w)|&\geq \left|A(w_A)-B(w_A)\right|\geq |A(w_A)-B(w_B)|\\
 &=\left|\max_{w\in \{0,1\}^n} A(w)-\max_{w\in \{0,1\}^n} B(w)\right|.
\end{align}
\qed
\end{proof}

We are now ready to prove the main result of this section, which is also stated as \cref{thm:oracle}. 
\begin{lemma}
    Let $q$ be a function on the security parameter $n$ and let $\epsilon\leq 1$ be a function on $q$ and $n$. Let $Q$ be a quantum unitary oracle.  Let $U\gets \mathcal{U}$ and $U'\gets \mathcal{U}'$ be two $n$-qubit unitaries sampled from some distributions $\mathcal{U}, \mathcal{U}'$ of unitaries. Assume that the following classical witness aided indistinguishability condition is satisfied: For any  $q$-time quantum algorithm $\mathcal{B}$,
\begin{align}
\underset{\begin{subarray}{c} U\gets \mathcal{U}\\
U'\gets \mathcal{U}'\end{subarray}}{\mathbb{E}}\left[ \max_{w\in \{0,1\}^{q}}  \left( \left|\Pr[\mathcal{B}^{Q,U}(1^n,w)=1]-\Pr[\mathcal{B}^{Q,U'}(1^n,w)=1]\right|\right)\right]
    \leq \epsilon(n,q)
\end{align}
for large enough $n$. 
    Then, for any $U,U'$, there exist oracles $\QCMA_U\in \mathcal{O}_{\textsf{PromiseQCMA}}^{Q,U}$ and $\QCMA_{U'}\in \mathcal{O}_{\textsf{PromiseQCMA}}^{Q,U'}$, such that for any  $q$-time quantum algorithm $\adv$,
    \begin{align}
\underset{\begin{subarray}{c} U \gets \mathcal{U}\\
U'\gets \mathcal{U}'\end{subarray}}{\mathbb{E}}\left[   \left|\Pr[\adv^{Q,U,c\QCMA_U}(1^n)=1]-\Pr[\adv^{Q,U',c\QCMA_{U'}}(1^n)=1]\right|\right]\\
    \leq 2q\cdot \epsilon(n,1000\cdot n\cdot q)+4q\cdot O(2^{-n})
\end{align}
    for large enough $n$.
\end{lemma}

\begin{proof}
We specify the oracles $\QCMA_U$ and $\QCMA_{U'}$ later in the proof. 

 Note that since oracle access to $\QCMA_{U}$ is classical, we can split the oracle queries so that every query is either to $\QCMA_{U}$ or to $(Q,U)$, but not both at the same time. This modification increases the bound on the number of queries to $2q$. We now define the following hybrids. Let $\Hy_0$ be the following experiment:
    \begin{itemize}
        \item Sample $U\gets \mathcal{U}$.
        \item Run $\adv^{Q,U,c\QCMA_U}(1^n)$ and output the result. 
    \end{itemize}
Inductively define the hybrid $\Hy_i$ for $i\in [2q]$ to be the same as $\Hy_{i-1}$ except the $i^{th}$ query to $(Q,U,c\QCMA_U)$ is replaced with $(Q,U',c\QCMA_{U'})$.


We now bound the trace distance between the hybrids $\Hy_{0}$ and $\Hy_{1}$. Consider the case where that the first query is to $\QCMA_U$. We can simulate the oracle $\QCMA_{U}$ with the oracle $\Sim^{Q,U} $ which is a quantum algorithm that queries ${U}$ at most $1000\cdot n\cdot q$ times and gets a classical witness $w$ of length at most $q$ to aid in responding to its query.  

In particular, $\Sim^{Q,U}$ acts as follows. Assume $\adv$ submits a query $(V,x)$ to $\Sim^{Q,U}$. $\Sim$ receives a witness $w$ and then, $\Sim$ computes $V^{Q,U}(x,w)$ $1000\cdot n$ times and if the average of these computations is less than $1/2$, it outputs 0 and outputs 1 otherwise. Note that since $\adv$ has run-time $q$, the length of $(V,x)$ is bounded by $q$ and, thus, the number of queries by $V$ is bounded by $q$, giving a bound of $1000\cdot n\cdot q$ on the total number of queries in this computation. 

Let $\QCMA_{U}$ be the algorithm in $\mathcal{O}_{\textsf{PromiseQCMA}}^{Q,U}$ that if queried on a non-promise input $(\tilde{V},\tilde{x})$, outputs $\Sim^{Q,U}(\tilde{V},\tilde{x},\tilde{w})$ where $\tilde{w}$ is set to the value that maximizes the probability of this computation returning 1. The output distributions differ on the promise inputs, but we can bound this error using the Hoeffding bound, obtaining:
\begin{align}
    \left|\max_{w}\Pr[\Sim^{Q,U}(V,x,w)=1]-\Pr[\QCMA_{U}(V,x)=1]\right|\leq O(2^{-n}).
\end{align}

Let $\QCMA_{U'}$ be the algorithm in $\mathcal{O}_{\textsf{PromiseQCMA}}^{Q,U'}$ that performs the same computation as $\Sim^{Q,U'}$ on non-promise inputs. We can similarly argue that oracle access to $\Sim^{Q,U'}$ simulates oracle access to $\QCMA_{U'}$ {with error $O(2^{-n})$.} 

Note that in hybrids $\Hy_{0}$ and $\Hy_{1}$, the first query is sampled from the same distribution, say $\textsf{D}_\adv$. 


For any query $(V,x)\gets \textsf{D}_\adv$, 
\begin{align}
&\underset{\begin{subarray}{c} U \gets \mathcal{U}\\
U'\gets \mathcal{U}'\end{subarray}}{\mathbb{E}}\left[ \left|\Pr\left[\QCMA_{U}(V,x)=1\right]-\Pr\left[\QCMA_{U'}(V,x)=1\right]\right|\right] \\
 \leq &\underset{\begin{subarray}{c} U \gets \mathcal{U}\\
U'\gets \mathcal{U}'\end{subarray}}{\mathbb{E}}\left[\left|\Pr\left[\QCMA_{U}(V,x)=1\right]-\max_{w}\Pr[\Sim^{Q,U'}(V,x,w)=1]\right|\right]+O(2^{-n})\\
 \leq &\underset{\begin{subarray}{c} U \gets \mathcal{U}\\
U'\gets \mathcal{U}'\end{subarray}}{\mathbb{E}}\left[\left|\max_{w}\Pr[\Sim^{{Q,U}}(V,x,w)=1]-\max_{w}\Pr[\Sim^{Q,U'}(V,x,w)=1]\right|\right] +2\cdot O(2^{-n})\\ \label{eq:bbb}
        \leq &\underset{\begin{subarray}{c} U \gets \mathcal{U}\\
U'\gets \mathcal{U}'\end{subarray}}{\mathbb{E}}\left[\max_{(\tilde{V},\tilde{x},\tilde{w})}\left(\left|\Pr[\Sim^{Q,{U}}(\tilde{V},\tilde{x},\tilde{w})=1]-\Pr[\Sim^{Q,U'}(\tilde{V},\tilde{x},\tilde{w})=1]\right|\right)\right]+2\cdot O(2^{-n})\\
  \label{eq:hh}  \leq &\epsilon(n,1000\cdot n\cdot q)+2\cdot O(2^{-n}).
\end{align} 
Here, \cref{eq:bbb} follows because of \cref{lem:bb}.
While \cref{eq:hh} follows from the theorem assumption by setting the witness to $(\tilde{V},\tilde{x},\tilde{w})$, which is of length at most $q$.
Therefore, this bounds the probability that $\adv$ can distinguish oracle access to $\QCMA_U$ and $\QCMA_{U'}$ from the first query, giving
\begin{align}
   \left|\Pr[\Hy_{0}=1]-\Pr[\Hy_{1}=1]\right|  \leq \epsilon(n,1000\cdot n\cdot q)+2\cdot O(2^{-n}).
\end{align}
If the first query was to $(Q,U)$, then indistinguishability follows directly from the assumption, which gives:
\begin{align}
   \left|\Pr[\Hy_{0}=1]-\Pr[\Hy_{1}=1]\right|  \leq \epsilon(n,q).
\end{align}
We can run the same argument to obtain the same bound of $\epsilon(n,1000\cdot n\cdot q)+2\cdot O(2^{-n})$ between any consecutive pair of hybrids. 

By the triangle inequality, we get 
\begin{align}
&\underset{\begin{subarray}{c} U \gets \mathcal{U}\\
U'\gets \mathcal{U}'\end{subarray}}{\mathbb{E}}\left[   \left|\Pr[\adv^{Q,U,c\QCMA_U}(1^n)=1]-\Pr[\adv^{Q,U',c\QCMA_{U'}}(1^n)=1]\right|\right]\\
&=\left|\Pr[\Hy_{0}=1]-\Pr[\Hy_{2q}=1]\right|  \\  &\leq 2q\cdot \epsilon(n,1000\cdot n\cdot q)+4q\cdot O(2^{-n}).
\end{align}
\qed
\end{proof}

\end{document}